\documentclass[reqno]{amsart}
\usepackage[colorlinks]{hyperref}
\usepackage{amssymb}
\usepackage{latexsym}
\usepackage{amstext} 
\usepackage{color}
\usepackage{amsthm}  
\usepackage{amsmath}
\newtheorem{Theorem}{Theorem}
\newtheorem{Lemma}[Theorem]{Lemma}
\newtheorem{Remark}[Theorem]{Remark}
\newtheorem{Definition}[Theorem]{Definition}

\newtheorem{Corollary}[Theorem]{Corollary}
\theoremstyle{example}
\newtheorem{Example}[Theorem]{Example}
\usepackage{geometry}
\usepackage[utf8]{inputenc}
\title{Uncertainty Principles in Krein Space}
\author{Sirous Homayouni and Angelo B. Mingarelli} 
\address{School of Mathematics and Statistics \\ School of Mathematics and Statistics \\ Carleton University \\ Ottawa, Ontario, Canada}
\thanks{}
\date{}
\email[S. Homayouni]{chomayou@mathstat.yorku.ca}
\email[A. B. Mingarelli]{angelo@math.carleton.ca}

\begin{document}
\begin{abstract}
Uncertainty relations between two general non-commuting self-adjoint operators are derived in a Krein space. All of these relations involve a Krein space induced fundamental symmetry operator, $J$, while some of these generalized relations involve an anti-commutator, a commutator, and various other nonlinear functions of the two operators in question. As a consequence there exist classes of non-self-adjoint operators on Hilbert spaces such that the non-vanishing of their commutator implies an uncertainty relation. All relations include the classical Heisenberg uncertainty principle as formulated in Hilbert Space by Von Neumann and others. In addition, we derive an operator dependent (nonlinear) commutator uncertainty relation in Krein space.
\end{abstract}
\maketitle
%***********************************************************************************************************************************************

\section{Introduction}

There are some quantum field theories (QFT) that, without the help of indefinite inner product spaces, face some inner contradictions. For example, these appear in Quantum Electrodynamics (in the Gupta-Bleuler formalism), vector meson theories and Pauli-Villar regularization procedure. There are some other QFTs, like Heisenberg's unified field theory of elementary particles, that are founded on the basis of indefinite metric spaces from the beginning in order to prevent divergences that usually emerge in them. The indefinite inner product makes it possible, by virtue of its non-positive definiteness, to make convenient subtraction procedures that eventually can make things like propagators finite. These comments provide some background as to why indefinite inner product spaces, sometimes also called {\it indefinite metric spaces}, play a useful and important role in QFTs and mathematics.

Along with this background there are tendencies to extend the theories regarding general eigenvalue problems of Sturm-Liouville type to the wider spaces with indefinite metric. Examples of these can be found in the works of H. Langer \cite{langer44} and A. B. Mingarelli \cite{A.B.Mingarelli44}, \cite{ABM86} on so-called {\it non-definite Sturm-Liouville problems}.

The spaces called {\it Pontryagin spaces} and {\it Krein spaces} were named in honor of the Soviet mathematicians, Lev Semenovich Pontryagin and Mark Grigorievich Krein, the first researchers who investigated the theory of indefinite inner product spaces \cite{Pontryagin244}. While such spaces endowed with indefinite inner products are not, strictly speaking, Hilbert spaces they admit a decomposition into a direct sum of two Hilbert spaces (usually called positive and negative spaces) \cite{Mingarelli44}. If either one of the two Hilbert spaces in the stated decomposition is finite dimensional, we call the original indefinite inner product space a Pontryagin space. Otherwise, the space is called a Krein space. This decomposition is done via {\it orthoprojectors}; these project the vectors of our space into components belonging to either the positive or negative Hilbert spaces and it is by means of these orthoprojectors that the most characteristic notion of these indefinite metric spaces is defined. We are referring here to the notion of a {\it fundamental symmetry operator}. In general, for any Pontryagin/Krein space there is a \textit{bona fide} inner product, defined via such a fundamental symmetry operator with the property that the indefinite space endowed with this new inner product is actually a Hilbert space, \cite{Azizov44}.

There are clear difficulties connected with an indefinite inner product space. For example, according to probabilistic interpretations of quantum theories the results of experiments are explained on the basis of probabilities, \cite{de broglie44}. Such probabilistic interpretations become awkward in connection with an indefinite inner product space since now some vectors can have zero or negative {\it norm} thus giving rise to zero or even negative probabilities. One way out of this problem is to decide which part of the state vector space (the base space with the indefinite metric) may describe actual physical states and then look for the possibility of an interpretation, \cite{Nagy44}.

The terminology associated with various types of vectors in indefinite metric spaces is colorful indeed. For example, those non-zero vectors having negative or zero norm are called {\it ghosts}. Since, as related above, such ghosts naturally lead to difficulties in terms of probabilistic interpretation, a decomposition of such vectors (or {\it states}) into physical and non-physical subspaces $H_p$ and $H_n$ respectively could be helpful. The decompositions of the state space for actual theories are based on physical reasonings and are usually called {\it natural decompositions} (with a norm that is at least ``semi-definite"). We recall that these decompositions are accomplished via projection operators $P_{p}$ and $P_{n}$ which project indefinite states into physical and non-physical subspaces $H_{p}$ and $H_{n}$, respectively. Such decompositions may then help to approximate the real world more closely.

In this work we will use the terms {\it space with an indefinite metric} and {\it indefinite inner product space} interchangeably and without notice. This choice of terminology has now become universal. Since a vector in an indefinite metric space could have positive, negative or zero norm, we may get problems with the completeness of the eigenvectors of a {\it hermitian} operator in such spaces (``hermiticity`` is now defined in terms of the indefinite inner product).  It turns out that the existence of at least one non-zero eigenvector with zero norm is a necessary condition for the eigenfunctions of a hermitian operator {\it not} to form a complete system. The important question is then: ``When do the eigenvectors of an (hermitian) operator $P$ form a complete system?" In indefinite inner product spaces the notion that is fundamental to the completeness of a set of eigenvectors revolves around the notion of a {\it principal linear manifold}. The principal linear manifold of an operator $P$ corresponding to a complex number $p$ is the linear span of the set of states $u$ for which there is a positive integer $l(u)$ such that $(P-pI)^{l(u)}u=0$ (i.e., $u$ is a generalized eigenvector).

While in Hilbert space the eigenfunctions of a self-adjoint operator with {\it pure point spectrum} (i.e., its spectrum consists only of eigenvalues having finite multiplicity) always form a complete set (i.e., one can expand any function in the Hilbert space in terms of them), the same is not necessarily true in a Krein/Pontryagin space. Thus, the notion of an eigenfunction is not enough and the more general concept of ``generalized eigenfunction" is needed.

In Section 2 we review the basic theory of Pontryagin and Krein spaces and their operators. The concepts, definitions and fundamental theorems of Pontryagin and Krein spaces are reviewed there.

In Section 3 a brief history of the Heisenberg Uncertainty Principle (or Relation, HUR) between two non-commuting self-adjoint operators is reviewed. When applied to the position and momentum operators this fundamental result lies at the core of quantum mechanics. In this form it was published first by Werner Heisenberg in 1927 \cite{Heisenberg44}. In the sequel we sometimes denote this principle by the acronym, HUR, for simplicity. Its abstract form, that is, its $L^2$-form, was originally proved by H. Percy Robertson in 1929, \cite{Robertson44}, for any two non-commuting self-adjoint operators in Hilbert space. This result was followed by another completely abstract Hilbert space version of HUR due to John Von Neumann in 1930, \cite{Neumann44}, on the basis of the former proof of Robertson. In 1930, Erwin Schr\"{o}dinger considered the fact that product of two non-commuting hermitian operators in general is not hermitian but can be split into a hermitian and a ``skew-hermitian" ($i \times$ hermitian) part. Based on this consideration, Schr\"{o}dinger in  \cite{Schrodinger44} derived a generalized version of HUR involving both a commutator part (that already appeared in Robertson's derivation) and a new anti-commutator part. While Schr\"{o}dinger derived this inequality for two general self-adjoint operators, he showed that for the specific case of the canonical conjugate operators momentum and position, the noncommutative term is not generally zero. Neglecting it, however, gives the weaker inequality of Heisenberg.

The questions we consider are inspired by a paper of Condon, \cite{condon44}. In this insightful paper Condon ponders the ``limits" of the then recently discovered uncertainty principle. In particular he states: ``The fact that the operators corresponding to two physical quantities, $p$ and $q$, do not commute does not imply the existence of an uncertainty relation of the form $$\triangle p\triangle q > \frac{h}{2\pi},$$ namely, that the product of two uncertainties must be greater than or equal to some lower limit." He goes on to show that there exist operators for which non-commutativity does not imply uncertainty.

This work is more concerned with the question:`` Can we extend the class of operators, $A$, $B$ and the spaces they are defined upon so as to {\it guarantee} an uncertainty relation?" We show that, indeed, there exist classes of non self-adjoint operators on Hilbert spaces such that the non-vanishing of their commutator implies an uncertainty relation.

In Section 4 we derive such a generalization of Heisenberg's Uncertainty Principle. While all known derivations and/or generalizations of Heisenberg's (classical) inequality are done for operators that are self-adjoint in Hilbert space, in this section we derive it for operators that are self-adjoint only in Krein space, thereby producing an uncertainty relation for a wider spectrum of operators that are not necessarily self-adjoint in Hilbert space. In the specific case of a Hilbert space this generalization reduces to the Robertson-Schr\"{o}dinger  inequality. We define a class of functions for which the anti-commutative part, called the {\it standard deviations product}, does not vanish thereby giving  a stronger/better lower bound than Heisenberg's original one. Also, under a modification on the commutator of the distance operator, $x$, and the momentum operator, $p$, this generalization can give us an expression of an uncertainty relation consistent with the minimum Planck length.

The applications and consequences of the generalized HUR in Krein space are discussed in the various sections where generalizations are presented. Conclusions are summarized in Section ~\ref{ch6}.

%***************************************************************************************************************************************************

\section{The basic theory of Krein spaces}

In this section the fundamental concepts and theorems of indefinite inner product spaces, Pontryagin, and Krein spaces are reviewed, \cite{Mingarelli44}.

\begin{Definition}\label{def1}
 Let $V$ be a vector space over the field of complex numbers $\mathbb{C}$. A {\it sesquilinear hermitian form} on $V$ is a map $Q:V\times V\rightarrow\mathbb{C}$ such that for all $x, y, x_{1},x_{2}\in\, V$ and $ \lambda_{1},\, \lambda_{2}\, \in\,\mathbb{C}$ we have

I) $Q(\lambda_{1}x_{1}+\lambda_{2}x_{2},y)=\lambda_{1}Q_{1}(x_{1},y)+\lambda_{2}Q_{2}(x_{2},y)$ (linearity in the 1st argument),

II) $Q(y,x)= \overline{Q(x,y)}$ (hermiticity).
\end{Definition}

\begin{Example} Let $V=L^{2}(a,b)$ where $-\infty<a<b<+\infty$ be the space of all complex-valued Lebesgue square-integrable functions on the real interval $[a,b]$, under the usual operations of sums, etc. Define a map $Q$ on $V\times V$ by
\begin{equation}
Q(f,g)=\int_{a}^{b}f(x)\overline{g(x)}\,dx.\nonumber
\end{equation}
Then $Q$ is clearly sesquilinear hermitian on $V$.
\end{Example}

\begin{Example} Let $r\in L^{\infty}(D)$ be a real valued essentially bounded function defined in a compact subset $D\subset \mathbb{R}^{n}$. The vector space $V=L_{r}^{2}(D)$ is defined to be the space of all those complex valued Lebesgue measurable functions $f$ such that $\int_{D}|f(x)|^{2}|r(x)|\,dx <\infty$, where $dx$ is Lebesgue measure in $\mathbb{R}^{n}$. Define a map by $Q$ on $V\times V$ by
\begin{equation}
Q(f,g)=\int_{D}f(x)\overline{g(x)}r(x)\,dx.\nonumber
\end{equation}
Then $Q$ is a sesquilinear hermitian form on $V$. This space $V$ is called a {\it weighted} Lebesgue space of square-integrable functions.
\end{Example}

The difference between a sesquilinear hermitian form and an inner product is that an inner product is a sesquilinear hermitian form with the additional property of {\it positive definiteness}, i.e., $Q(x,x)>0$ for all $x\neq 0$. Therefore we can say that a sesquilinear hermitian form generally defines an {\it indefinite} inner product. Also a sesquilinear hermitian form is sometimes called (especially in Russian papers) an \emph{indefinite metric} (although strictly speaking it is not an indefinite metric as the triangle inequality may be lacking). Because of the pervasive nature of these equivalent expressions we will use the terms {\it indefinite metric space} and {\it indefinite inner product space} interchangeably.

\begin{Definition} Let $V$ be a vector space with an indefinite metric. We say that a vector $x \in V$ is [1]
\begin{equation}
\left\{
  \begin{array}{ll}
     &\hbox{{\rm positive, if}\,\, $[x,x]>0$;} \\
     &\hbox{{\rm negative, if}\,\, $[x,x]<0$;} \\
     &\hbox{{\rm neutral, if}\,\, $[x,x]=0$, $x\neq 0$.}
  \end{array}
\right.
\end{equation}
\end{Definition}

\begin{Example} Let $u=(u_{1},u_{2},u_{3,}u_{4}), v \in\mathbb{R}^4$. Then it is readily verified that $$[u,v]=u_{1}v_{1}-u_{2}v_{2}-u_{3}v_{3}-u_{4}v_{4}$$ is a sesquilinear hermitian form on $\mathbb{R}^4$. It is also clear that $[u,v]$ could be positive, negative or zero depending on the vectors chosen. Hence $[, ]$ defined here is an indefinite metric, and this metric is usually called the {\it Lorentz metric}.
~\label{exam3}
\end{Example}

\begin{Example} Consider the group $G$ of all invertible $4\times4$ matrices $A$ over the real numbers such that $[Au,v]=[u,A^{-1}v], \forall u, v\in \mathbb{R}^4$ where $[\ ,\ ]$ is the Lorentz metric of Example ~\ref{exam3}. This group $G$ is called the {\it Lorentz group} and is denoted by $O(3,1)$. Letting $v=Au$, we come up with $[Au,Au]=[u,u]$ so that the group $G$ leaves the form $[\ ,\ ]$ invariant. This also means that $G$ leaves the quadratic form $[u,u]=u_{1}^{2}-u_{2}^{2}-u_{3}^{2}-u_{4}^{2}$ invariant. The elements of the Lorentz group are called {\it Lorentz transformations}. For example, the matrix

\[ \left( \begin{array}{cccc}
\cosh\varphi & \sinh\varphi & 0 & 0 \\
\sinh\varphi & \cosh\varphi & 0 & 0 \\
0 & 0 & 1 & 0  \\
0 & 0 & 0 & 1 \end{array} \right)\]

where $\sinh\varphi=\frac{v\gamma}{c}$, $\cosh\varphi=\gamma=\frac{1}{\sqrt{1-\frac{v^{2}}{c^{2}}}}$ where $v$ is the velocity here, $c$, the speed of light, is a parameter of the group $G$. This represents a Lorentz transformation.
\end{Example}

\begin{Example}\label{expl5}
 A generalized matrix eigenvalue problem may be defined as an eigenvalue problem of the form
 \begin{equation}Ax=\lambda Bx\end{equation} where $A=A^*$, $B=B^*$ are hermitian matrices and we seek $\lambda \in \mathbb{C}$ such that $det(A-\lambda B)=0$ for some non-zero vector $x$ (usually called a {\it generalized eigenvector}). If we assume that $A$ is a positive definite matrix (i.e., $(Au,u)>0$ for all $u \in \mathbb{C}^n$ where $(,)$ is the usual inner product on $\mathbb{C}^n$), then $(Ax,x)=\lambda(Bx,x)$ and $\lambda$ must be real.
 Such a matrix $B$ can be used to define a sesquilinear hermitian form $[\ ,\ ]$ on $\mathbb{C}^n$ by setting $$[u,v]=(Bu,v)$$ where $(,)$ is the usual inner product on $\mathbb{C}^n$.
\end{Example}

 \begin{Remark} The assumption of positive definiteness on $A$ in the above example guarantees that any generalized eigenvalue $\lambda$ must be real. This can be easily seen because $(Ax,x)\neq 0$ for any $x\neq 0$. Now, the hermiticity of both $A$ and $B$ implies that both sides of $(Ax,x)=\lambda (Bx,x)$ are real. Taking imaginary parts we get $({\rm Im}\,\, \lambda)\, (Bx,x) =0$. If ${\rm Im}\,\lambda \neq 0$ the latter implies that  $(Bx,x) =0$. Thus $(Ax,x)=0$ which finally forces $x=0$ (since $A$ is positive definite). Therefore ${\rm Im}\, \lambda=0$ if $x$ is not zero, thus all generalized eigenvalues must be real whenever $A$ is positive definite.
 \end{Remark}

 \begin{Remark} It follows that any generalized eigenvalue $\lambda$ and associated generalized eigenvector $x$ must satisfy  $$\lambda\,(Bx,x)>0,$$ (if $A$ is positive definite). So, in this case, $(Bx,x)$ always has the same sign as $\lambda$. Using the indefinite inner product defined in Example~\ref{expl5} we deduce that,
\begin{equation}
\left\{
  \begin{array}{ll}
     & \hbox{$\lambda>0$ if and only if $[x,x]>0$;} \\
     & \hbox{$\lambda<0$ if and only if $[x,x]<0$;} \\
     & \hbox{$[x,x]=0$, if and only if $x=0$}
  \end{array}
\right.\nonumber
\end{equation}
\end{Remark}

 \begin{Definition}{$Q-$orthogonality} Let $Q(u,v)=[u,v]$ be a sesquilinear hermitian form on a vector space $V$. If for some $u,v \in V$ we have $[u,v]=0$, we say $u$ and $v$ are {\it orthogonal} with respect to $[\ ,\ ]$, or $u$ is $Q-$orthogonal to $v$.
 \end{Definition}
 \begin{Remark} Note that since $[\ ,\ ]$ is hermitian, $[u,v]=0 \Leftrightarrow [v,u]=0$ so that Q-orthogonality is a reflexive relation.
\end{Remark}
Non-zero neutral vectors are orthogonal to themselves, and the existence of such vectors in indefinite metric spaces is guaranteed by the following theorem.

\begin{Theorem} (See, e.g., \cite{ikl}, \cite{Mingarelli44}.)
Let $[\ ,\ ]$ be an indefinite sesquilinear hermitian form on $V$, in the sense that there are at least two vectors $x,y\in V$ such that $[x,x]>0$ and $[y,y]<0$. Then $V$ contains at least one (non-zero) neutral vector.
\end{Theorem}

As an application of the above theorem, considering the Lorentz metric, the mere existence of a time-like vector (our space) and the existence of a space-like vector (future) together implies the existence of a light-like vector (a particle moving with the speed of light) that is we can infer the existence of photons from the existence of tachyons (super-luminal particles).

\begin{Remark}
It can be easily shown that the eigenvectors corresponding to non-complex-conjugate generalized eigenvalues (see Example~\ref{expl5})
are $Q-$orthogonal. In other words, if $Ax=\lambda Bx$ and $Ay=\mu By$ where $\lambda, \mu \in \mathbb{C}$ and $x,y\neq 0$, then $[x,y]= (Bx,y)=0$ if $\lambda \neq \overline{\mu}$. Since $(Ax,y)=\lambda (Bx,y)$ together with $(Ax,y)=(x,Ay)=\overline{(Ay,x)}=\overline{(\mu By,x)}=\overline{\mu}\overline{(By,x)}=\overline{\mu}(x,By)=\overline{\mu}(Bx,y)$ implies $\lambda(Bx,y)=\overline{\mu}(Bx,y)$ $\Rightarrow$ $(\lambda-\overline{\mu})(Bx,y)= 0$. Therefore
\begin{equation}
 (Bx,y)=[x,y]= 0 \quad if\quad \lambda \neq\overline{\mu}.\nonumber
\end{equation}
\end{Remark}
\begin{Remark} The eigenvalue $\lambda$ in the preceding remark above cannot be non-real if $A$ is positive definite, (since all the eigenvalues must be real in this case). However, if $A$ is indefinite then there may well be real eigenvalues $\lambda$ whose eigenvectors $x$ satisfy $[x,x]=0$, \cite{A.B.Mingarelli44}.
\end{Remark}

When both matrices $A$ and $B$ are indefinite (i.e., their quadratic forms $(Ax,x)$, $(Bx,x)$ are indefinite) things can get pretty bad in the sense that there may be examples where every $\lambda\in\mathbb{C}$ is an eigenvalue of $Ax=\lambda Bx$. We choose $A, B, x$ as follows:
\begin{center}
$x=\left(\begin{array}{c}
    0 \\
    1 \\
    0
  \end{array}\right)$,
  $A=\left(\begin{array}{ccc}
    1 & 0& 0\\
    0 & 0& 0\\
    0 & 0& -2
  \end{array}\right)$,
  $B=\left(\begin{array}{ccc}
    -1 & 0& 0\\
    0 & 0& 0\\
    0 & 0& 2
  \end{array}\right)$.
\end{center}
 Then $Ax=0, Bx=0$ and so $Ax=\lambda Bx$ for all $\lambda \in \mathbb{C}$. Of course, such is the case since $x\neq 0$ and $x\in ker(A) \cap ker(B)$.

 In this example the matrices $A$ and $B$ are not invertible (as Ker$(A)$ and Ker$(B)\neq \{0\}$). On the other hand, using the arguments above, one can show that if for some matrices $A$ and $B$, all of the eigenvalues of the generalized eigenvalue problem $Ax=\lambda Bx$ are complex, then both $A$ and $B$ must be indefinite (in fact, we can state that neither $A$ nor $B$ can be invertible in this case).

 Thus, in order to prevent all the eigenvalues of the generalized eigenvalue problem $Ax=\lambda Bx$ from filling the whole complex plane, we need that at least one of the matrices $A$ and $B$ be invertible. Suppose the invertible one is $B$. Then
\begin{equation}
Ax=\lambda Bx \Leftrightarrow B^{-1}Ax=\lambda x,\nonumber
\end{equation}
i.e., we come up with a standard eigenvalue problem for the matrix $C=B^{-1}A$ (which is not necessarily hermitian even though $A$ and $B$ are).

One of the questions that arises in this context is the following one: ``How do we extend the space in such a way that the product $C=B^{-1}A$ is again hermitian?" We address this question by defining a generally indefinite inner product $[\ ,\ ]$ by setting
\begin{equation}[u,v]:=(Bu,v)
\end{equation}
where $(,)$ is the usual inner product of $\mathbb{C}^n$. We know that this defines a sesquilinear hermitian form (as $B$ is hermitian). Furthermore, for any $u, v$,
\begin{equation}[Cu,v]=(BCu,v)=(Au,v)=(u,Av)=(u,BCv)=(Bu,Cv)=[u,Cv],
\end{equation}
i.e., we note that now $C=B^{-1}A$ is ``hermitian`` since $[Cu,v]=[u,Cv]$, relative to this new inner product $[\ ,\ ]$.

Next, if $\lambda$ is a non-real eigenvalue of $C$ and $u$ is a corresponding eigenvector then, for our inner product $[\ ,\ ]$,
$$[Cu,u]=(BCu,u)=\lambda(Bu,u)=\lambda[u,u],$$ i.e.,
\begin{equation}
[Cu,u]=\lambda [u,u].
\end{equation}

Since $C=B^{-1}A$ is hermitian relative to the sesquilinear hermitian form $[\ ,\ ]$, the quantity $[Cu,u]$ must be real, i.e., the (generalized) {\it expectation value}, $[Cu,u]$, of $C$ relative to $[\ ,\ ]$ is real. This, together with (2.5) implies that
\begin{equation}
  {\rm Im}\,[Cu,u]=0= {\rm Im}\,(\lambda)[u,u] \Rightarrow [u,u]=0,
  \end{equation}
  since ${\rm Im}\, (\lambda)\neq 0$ (by hypothesis). Thus, the non-real eigenvalues of $C$ have necessarily neutral eigenvectors (relative to $[\ ,\ ]$).

  The remarks leading to (2.3) motivate the following definition of hermitian operators (relative to a sesquilinear hermitian form $[\ ,\ ]$).

  \begin{Definition} Let $[\ ,\ ]$ be a sesquilinear hermitian form on a vector space $V$. We say that a linear transformation $A$ on $V$ is hermitian if for every $ u, v\in V$ we have,
  \begin{equation}
  [Au,v]=[u,Av].
  \end{equation}
 \end{Definition}

 \begin{Remark} In the case when $B$ is an $n\times n$  positive definite matrix, $[\ ,\ ]$ defined by (2.3), is just a second inner product on the space.
 \end{Remark}

 \begin{Remark} In the case of a real hermitian operator on $(\mathbb{R}^n,(,))$ we know from linear algebra that $A=A^*$ $\Leftrightarrow$ $A=A^t$.
 \end{Remark}

 %*********************************************************************************************************************

\section{Decomposition into positive and negative subspaces}

Let $B$ be an $n \times n$ hermitian matrix such that the inner product $[u,v]:=(Bu,v)$ is non-degenerate (i.e., if for all $v$ we have $[u,v]=0$, then $u=0$). The indefinite inner product $[u,v]:=(Bu,v)$ can be used as a basis for decomposing the indefinite inner product space into a {\it Q-orthogonal direct sum} of so-called positive and negative subspaces. We proceed by way of an example.

\begin{Example}
Consider the indefinite inner product $[u,v]:=(Bu,v)$ on $\mathbb{R}^2$ where
$$B=\left(\begin{array}{cc}
1 & 0 \\
0 & -1
\end{array}\right)\in \mathbb{R}^{2\times2},$$

$u= {\rm col}\  (u_1, u_2)$ and $v = {\rm col}\  (v_1,v_2)$.

Then $[u,v]:=(Bu,v)= u_1v_1-u_2v_2$ $\Rightarrow$ $[u,u]=u_1^2-u_2^2$. Now the usual basis in $\mathbb{R}^2$ is
$e_1= {\rm col}\  (1, 0)$, $e_2= {\rm col}\  (0, 1)$. Since $[e_1,e_1]=1$ and  $[e_2,e_2]=-1$,
it follows that $e_1$ and $e_2$ are respectively positive and negative vectors in $(\mathbb{R}^2,[,])$, and the subspaces generated by $e_1$ (respectively $e_2$) are positive (respectively negative) subspaces of $(\mathbb{R}^2,[,])$. These subspaces are denoted respectively by $H^+$ and $H^-$. Hence a vector $u$ in the indefinite inner product space $(\mathbb{R}^2, [,])$ has a standard decomposition as
$u=u_1e_1+u_2e_2$ where now $u_1e_1$ and $u_2e_2$ belong to $H^+$ and $H^-$ respectively. This direct sum is $Q$-orthogonal in the sense that $[e_1,e_2]=(Be_1,e_2)=0$, or the two subspaces $H^+$ and $H^-$ are orthogonal with respect to $[\ ,\ ]$. When $H^+$ and $H^-$ are $[\ ,\ ]$-orthogonal, we write the decomposition of the original space in the form $H^+ [+] H^-$.

Hence we have come to the decomposition of $\mathbb{R}^2$ as an indefinite inner product space, $(\mathbb{R}^2, [ , ])$, into an $[ , ]$-orthogonal direct sum $\mathbb{R}^2 = H^+ [+] H^-$ of subspaces $(H^+, +[ , ])$ and $(H^-,-[ , ])$. Observe that $[ , ]$ (respectively $-[ , ]$) is an inner product on $H^+$ (respectively $H^-$).

We also note that neutral vectors (they necessarily exist), i.e., vectors $u={\rm col}\  (u_1,u_2)$ such that $u_1^2 = u_2^2$, or $|u_1| = |u_2|$
also have a decomposition into a $[\ ,\ ]$-orthogonal sum of two vectors each of which lies in the positive and negative spaces, $H^+$ and $H^-$.
\end{Example}

The following definitions lead us to the study of orthogonal complements.

\begin{Definition}
A vector $x_0$ is called an {\it isotropic vector} of the vector space $V$ if $x_0 \bot v$ ($x_0$ is orthogonal to $v$ relative to [,]) for all $v \in V.$
\end{Definition}

\begin{Definition} Let $\mathfrak{L}\subset(V,[,])$ be a subset of $V$. Then the [,]-orthogonal complement of $\mathfrak{L}$ is defined as the set of vectors $u\in V$ with $u\bot \mathfrak{L}$ in the sense of [,].
\end{Definition}
Thus $u \in \mathfrak{L}^\bot$ $\Leftrightarrow$  $[u,\mathfrak{L}]=0$.
Note that since the zero vector is $[,]-$orthogonal to all vectors, $\mathfrak{L}^\bot$ is a subspace, even when $\mathfrak{L}$ is merely a set, (see \cite{ikl}, \cite{Krein44}).

\begin{Definition}
Let $\mathfrak{L}$ be a subspace of $V$. The isotropic subspace $\mathfrak{L}^o$ of $\mathfrak{L}$ is the space
of all isotropic vectors of $\mathfrak{L}$.
\end{Definition}

\begin{Remark}
Note that $\mathfrak{L}^o=\mathfrak{L}\cap\mathfrak{L}^\bot$.
\end{Remark}
This is because if $u\in\mathfrak{L}^o$ then $u \in \mathfrak{L}$, by definition. Furthermore, $u\in\mathfrak{L}^o$ also implies that  $u \bot \mathfrak{L}$. Hence $u \in \mathfrak{L}^\bot$. Thus, $\mathfrak{L}^o \subseteq \mathfrak{L}\cap\mathfrak{L}^\bot$. Conversely, if $u \in \mathfrak{L}\cap\mathfrak{L}^\bot$, then $u\in \mathfrak{L}$ and $u\in \mathfrak{L}^\bot$. The latter implies that $[u, \mathfrak{L}] = 0.$ Hence $u$ is an isotropic vector of $\mathfrak{L}.$ Thus, $\mathfrak{L}\cap\mathfrak{L}^\bot \subseteq \mathfrak{L}^o$. The result follows.

%%%%%%%%%%%%%%%%%%%%%%%%%%%%%%%% Feb 2, 2011 %%%%%%%%%%%%%%%%%%%%%%%%%%

\begin{Example} The indefinite inner product space $(\mathbb{R}^2,[,])$ where $[u,v]=(Bu,v)$ and $$B=\left(
                    \begin{array}{cc}
                      1 & 0 \\
                      0 & -1 \\
                    \end{array}
                  \right)$$ has no isotropic vector other than the zero vector. This is because if we let
$x={\rm col}\  (x_1,x_2)$ be an isotropic vector in $(\mathbb{R}^2,[,])$ then, by definition, $[x,u]=0$ for all $u\in \mathbb{R}^2$.
In particular, choosing $u={\rm col}\  (1,0)$ then $0=[x,u]=(Bx,u)=x_1.$ Similarly, using $v={\rm col}\  (0,1)$ we see that $0=[x,v]=(Bx,v)=x_2.$ Thus $x=0$, i.e., there is no isotropic vector in $(\mathbb{R}^2,[,])$ other than the zero vector.
\label{exam7}
\end{Example}

\begin{Example} We use the matrix $$B=\left(
         \begin{array}{ccc}
           1 & 0 & 0 \\
           0 & -1 & 0 \\
           0 & 0 & 0 \\
         \end{array}
       \right)$$ to define an indefinite inner product, $[\ ,\ ]$, on $\mathbb{R}^3$ by setting $[u,v]=(Bu,v)$, for all $u,v\in \mathbb{R}^3$. It is easy to see that $[u,v] = u_1v_1-u_2v_2.$ Now let $\mathfrak{L}=\{ \alpha\, {\rm col}\  (0,0,1) : \alpha\, \in \mathbb{R}\}$, be the linear span of the vector $e_3 \in \mathbb{R}^3$.

       Observe that for any given $u\in \mathbb{R}^3$ we have $[u,v]=0$ for all  $v\in \mathfrak{L}$. It follows that $\mathfrak{L}^\bot=\mathbb{R}^3$ and so                    $\mathfrak{L}^o=\mathfrak{L}\cap\mathfrak{L}^\bot=\mathfrak{L}\cap\mathbb{R}^3=\mathfrak{L}$, i.e., $\mathfrak{L}^o=\mathfrak{L}$ i.e., the isotropic subspace of $\mathfrak{L}$ is $\mathfrak{L}$ itself. Since $\mathfrak{L}^o\neq{0}$, we say that $\mathfrak{L}^o$ (and so $\mathfrak{L}$) is a degenerate isotropic subspace and it is the subspace generated by $e_3$.
       \label{exam8}
                  \end{Example}
\begin{Remark}
It can be shown that, in this case, there is no decomposition of $\mathbb{R}^3$ into a $[,]-$orthogonal direct sum of definite spaces.
\end{Remark}

\begin{Definition} Let $[\ ,\ ]$ be an indefinite inner product on the infinite dimensional vector space $(V,[,])$ over $\mathbb{C}$. If $V$ admits a canonical decomposition $V=H^+[+]H^-$, where $(H^+,[,])$ and $(H^-,-[,])$ are Hilbert spaces, is called a Krein space. In the case where either $H^+$ or $H^-$ is finite dimensional, $V$ is called a {\it Pontryagin space}.
\end{Definition}
\begin{Example} While the non-degenerate space $(\mathbb{R}^2,[,])$ of Example~\ref{exam7} admits the $[,]-$ orthogonal decomposition $(\mathbb{R}^2,[,])=H^+ [+] H^-$, where $(H^+,[,])$ and $(H^-,-[,])$ are Hilbert spaces. Here $(H^+,[,])$ (resp. $(H^-,-[,])$) is defined by the span of the vector ${\rm col}\ (1,0)$ (resp. ${\rm col}\ (0,1)$). Thus, $(\mathbb{R}^2,[,])$ is a Pontryagin space. However, in the case of the degenerate space $(\mathbb{R}^3,[,])$ of Example~\ref{exam8} one can prove that because of the presence of the isotropic subspace spanned by ${\rm col}\ (0,0,1)$ such a decomposition is not admitted and so this cannot be a Krein (Pontryagin) space. The above conclusions for $(\mathbb{R}^2,[ ,])$ and $(\mathbb{R}^3,[ ,])$ are consequences of the more general Lemma below whose proof is basically clear from the definitions.
\end{Example}
\begin{Lemma}
If $(V,[,])$ is a Krein space then $V$ cannot contain a (non-trivial) isotropic vector.
\end{Lemma}
\begin{Corollary}
 The space $(\mathbb{R}^3,[ ,])$ with indefinite inner product defined as in Example~\ref{exam8}, namely, $[u,v]=(Bu,v)$ where $$B=\left(
       \begin{array}{ccc}
        1 & 0 & 0 \\
         0 & -1 & 0 \\
          0 & 0 & 0 \\
          \end{array}
           \right)$$
is not a Krein (Pontryagin) space as it contains an isotropic subspace.
\end{Corollary}
\begin{Example}
We show that the space $(\mathbb{R}^4,[ ,])$ with the indefinite inner product $[u,v]=(Bu,v)$ where $$B=\left(
       \begin{array}{cccc}
        1 & 0 & 0 & 0 \\
         0 & -1 & 0 & 0 \\
          0 & 0 & -1 & 0 \\
          0 & 0 & 0 & -1 \\
          \end{array}
           \right)$$
           is a Pontryagin space. \\ \\

Note that for $4-$vectors $u$ and $v$ we have
 \begin{equation}
[u,v]=u_1v_1-u_2v_2-u_3v_3-u_4v_4 \quad \Rightarrow \quad [u,u]=u_1^2-u_2^2-u_3^2-u_4^2,\nonumber
\end{equation}
which is the Lorentz metric.\\
Now consider the standard basis for $\mathbb{R}^4$ given by
\begin{equation}
e_1=(1,0,0,0)^t, e_2=(0,1,0,0)^t,\ldots .\nonumber
\end{equation}
Then $[e_1,e_1]=1$, $[e_2,e_2]=-1$, $[e_3,e_3]=-1$, and $[e_4,e_4]=-1$. In addition, $[e_1,e_i]=0$ for $i=2,3,4$ and $[e_i,e_j]=0$, $i,j=2,3,4$.
So $e_1$ is a positive vector and $e_2,e_3,e_4$ are negative vectors.

Then the subspace generated by $e_1$ (i.e., $H^+={\rm span}\,\{e_1\}$) is a positive subspace, while the subspace spanned by $e_2, e_3, e_4$ (i.e., $H^-= {\rm span}\,\{e_2,e_3,e_4\}$) is a negative subspace (relative to the indefinite metric, $[\,,\,]$).

%%%%%%%%%%%%%%%%%%%%%%%%%%%%%%% Feb.7

Then $u\in\mathbb{R}^4$ may be decomposed as $u=(u_1, u_2, u_3, u_4)^t= u_1e_1+u_2e_2+u_3e_3+u_4e_4$, where $u^+=u_1e_1\in H^+$ (as $[u^+,u^+]=u_1^2>0$)  and $u^-=u_2e_2+u_3e_3+u_4e_4\in H^-$ (as $[u^-,u^-]=-u_2^2-u_3^2-u_4^2<0$).

Furthermore,
             $[u^-,u^-]=0$ $\Rightarrow$ $u_2=u_3=u_4=0$ $\Rightarrow$ $u^-=0$ $\Rightarrow$ $(H^-,-[,])$ is a positive definite space and so is a Hilbert space .

             Similarly $[u^+,u^+]=0$ $\Rightarrow$ $u_1=0$ $\Rightarrow$ $u^+=0$ $\Rightarrow$ $(H^+,[,])$ is a positive definite space and so is a Hilbert space.

             On the other hand since $[e_1,e_i]=0$ for $i=2,3,4$ $\Rightarrow$ $[e_1, u^-]=0$ thus $[H^+,H^-]=0$ i.e., subspaces $H^+$ and $H^-$ are orthogonal relative to $[\ ,\ ]$.

             Therefore there is a decomposition of $\mathbb{R}^4$ as $\mathbb{R}^4=H^+[+]H^-$ being an orthogonal (relative to $[\ ,\ ]$) direct sum of two Hilbert spaces and the two Hilbert spaces are finite dimensional. Hence $(\mathbb{R}^4, [,])$ is a Pontryagin space.
\end{Example}
\begin{Remark}
In general one can show that $(\mathbb{R}^n, [,])$, with $[u,v]=(Bu,v)$, where $B$ is any invertible symmetric matrix, is a Pontryagin space.
\end{Remark}
\begin{Theorem}
Let $(V,[,])$ be a Krein space and $V=H^+[+]H^-$ its canonical decomposition. Define $(,):V\times V$$\rightarrow$$\mathbb{C}$ by \begin{equation}
\label{ind0}
(u,v)=[u^+,v^+]-[u^-,v^-]
 \end{equation}
 where $u=u^+ + u^-$, $v=v^+ + v^-$, $u^\pm, v^\pm \in H^\pm$. Then $(V,(,))$ is a Hilbert space (often called the {\bf Hilbert majorant} space of $(V,[,])$.
\end{Theorem}

\begin{Theorem}
Let $(V,[,])$ be a Krein space and $(V,(,))$ its Hilbert majorant. Then
\begin{equation}\label{ind1}
[u,v]=(u^+,v^+)-(u^-,v^-).
\end{equation}
Let $\|,\|$ denote the corresponding Hilbert space norm defined by $\|u\|^2=(u,u).$ Then,
\begin{equation}
\label{ind2}
[u,u]=\|u^+\|^2_V-\|u^-\|^2_V.
\end{equation}
\end{Theorem}
We can summarize the preceding discussion by saying that for any Krein space $(V,[,])$ there is an inner product $(,)$ on $V$ such that $(V,(,))$ is a Hilbert space. The relation between the definite, $(,)$, and indefinite, $[\ ,\ ]$, inner products is
\begin{equation}[u,v]=(u^+,v^+)-(u^-,v^-)\end{equation}
or its inversion
\begin{equation}\label{ind4}
(u,v)=[u^+,v^+]-[u^-,v^-]
\end{equation}
where $u^\pm, v^\pm \in H^\pm$ and $V=H^+[+]H^-$ is the canonical decomposition.

\section{The fundamental symmetry operator}

The notion of a {\bf fundamental symmetry operator} is one of the fundamental concepts in the study of Krein spaces. By way of background material we proceed by first introducing the concept of {\bf ortho-projectors}.

Let $(V,[,])$ be a Krein space and consider its orthogonal direct sum (canonical) decomposition $V=H^+[+]H^-$ (relative to $[,]$). Let $u \in V$, then $u=u^++u^-$ where $u^+$ (respectively $u^-$) belong to $H^+$ (respectively $H^-$).

For each $u\in V$ we define the two linear ortho-projectors $P_\pm$ by setting
$$P_+u=u^+,\quad\quad P_-u=u^-.$$
Clearly, $(P_++P_-)u=P_+u+P_-u=u^++u^-=u$ $\Rightarrow$ $P_++P_-=I$, where $I$ is the identity operator on $V$.
Thus, the operators $P_\pm$ project $u$ onto $H^\pm$ respectively, where each space $H^\pm$ is the orthogonal complement of the other (relative to $[,]$).

Associated with a given canonical decomposition $V=H^+[+]H^-$ of a Krein space $(V, [,])$, there is an operator $J:V\to V$, called a {\bf fundamental symmetry} that is defined in terms of the ortho-projectors $P_{\pm}$ by
\begin{equation}\label{cso}
J = P_+ - P_-.
\end{equation}
In other words, for $u\in V$ we have $Ju=u^+ - u^-$.
We now define the notion of an {\it adjoint}, or ``$J$-adjoin`` of an operator in a Krein space, $(V, [,])$, with Hilbert majorant, $(V, (,))$. Let $A$ be a linear operator whose domain, $D(A)$, is dense in $V$ and let $D(A^+)$ be the set of all vectors $v\in V$ such that there is an associated vector $z\in V$ such that for all $u\in V$ we have $$[Au,v] = [u,z].$$ Then $D(A^+)$ is a subspace of $V$ and the vector $z$ is uniquely determined since $V$ contains no isotropic vectors, by assumption. We usually write $z = A^+v$ where $A^+$ denotes the {\it Krein space adjoint} of the operator $A$. As usual we will maintain the notation $A^*$ for the Hilbert space adjoint of $A$. The (Krein space) adjoint therefore satisfies
$$[Au,v] = [u,A^+v],\quad\quad u\in D(A), v\in D(A^+).$$
The following properties of the Krein space adjoint hold (proofs may be found in \cite{ikl}) and are similar to the case of Hilbert space adjoints.

\begin{Lemma}\label{lem1} Let $(V, [,])$ be a Krein space with Hilbert majorant, $(V, (,))$. The Krein space adjoints, $A^+, B^+$, of densely defined operators $A, B$ satisfy
\begin{enumerate}
\item $A^+ + B^+ \subseteq  (A+B)^+ $
\item $B^+A^+ \subseteq (AB)^+   $
\item $A \subseteq (A^+)^+$
\item $A^* = JA^+J$
\item $A^+ = JA^*J$
\item $(\lambda A)^+ = \overline{\lambda} A^+,\quad \lambda \in \mathbb{C},$
\item If $A \subseteq B$ then $B^+ \subseteq A^+$,
\end{enumerate}
provided all the cited operators exist and have dense domains of definition.
\end{Lemma}
\begin{Remark} All inclusions in Lemma~\ref{lem1} become equalities in the case where the operators $A,B$ are each bounded on $(V, (,))$, although strict inclusions in (1)-(3) may arise when either $A$ or $B$ or both $A$ and $B$ are unbounded, for example, if either $A$ or $B$ is not closed. A sufficient condition for (3) to hold with equality is that $A$ be closed.
\end{Remark}

\begin{Lemma} \label{lem2} The fundamental symmetry operator $J$ defined in \eqref{cso} has the following properties:
\begin{enumerate}
\item For $u\in H^{+}$ (resp. $u \in H^-$), $Ju = u^+$ (resp. $Ju = -u^-$),
\item $J^2=I$, where $I$ is the identity operator on $V$,
\item $J$ is invertible and $J^{-1}=J$,
\item $J$ is a bounded operator (viewed as an operator on the Hilbert majorant space and its endowed norm) and $\|J\|=1$.
\item $J$ relates the inner products $(,)$ and $[,]$ on $V$ as follows: For all $u, v \in V$,
\begin{equation}\label{eq215}
(u,v) = [Ju,v]
\end{equation}
 and
 \begin{equation}\label{eq216}
[u,v] = (Ju,v).
\end{equation}
\item $J=J^*$ is self-adjoint (relative to the inner product $(,)$ on the majorant space),
\item $JJ^*=I$, i.e., $J$ is unitary (viewed as an operator on the Hilbert majorant space),
\end{enumerate}
\end{Lemma}

\begin{proof} We outline the proof. The first claim being clear we proceed to prove the second claim. For $u\in V$, $J^2u=J(Ju)=J(u^+-u^-)=Ju^+-Ju^-=u^+-(-u^-)=u^+ + u^-=u$. From this there now follows the third claim.

For $u\in V$ write  $\|u\|^2=(u,u)$. Since $Ju = u^+ - u^-$ where $(u^+,u^-) = 0$, we see that $$\|Ju\|^2 = (u^+ - u^-,u^+ - u^-) =(u^+,u^+)+(u^-,u^-)=(u^++u^-,u^++u^-)$$
$$=\|u^++u^-\|^2=\|u\|^2,$$ from which we infer that $J$ is a bounded operator on $V$ and, in fact, $\|J\|=1$.

The relationship between the inner products is straightforward. It is easy to see that \eqref{eq216} follows from \eqref{eq215} upon replacing $u$ by $Ju$ and using the fact that $J^2=I$. So we only have to verify  \eqref{eq215}.
This is also straightforward as
\begin{eqnarray*}
[Ju,v] &=& [u^+-u^-,v^++v^-]\\
&=& [u^+,v^+] - [u^-,v^-]\\
&=& (u,v)
\end{eqnarray*}
by \eqref{ind4}.

In order to prove the sixth claim, we first show that $J$ is a symmetric operator on the Hilbert majorant space. In order to prove symmetry we must show that $(Ju,v) = (u,Jv)$, for all $u, v, \in V$. But this is clear since
\begin{equation*}(Ju,v)=(u^+-u^-,v^++v^-)=(u^+,v^+)-(u^-,v^-)=(u^++u^-,v^+-v^-)=(u,Jv).
\end{equation*}
On the other hand, by definition of the adjoint of a bounded operator, we have for all $u, v \in V$,
\begin{equation}(Ju,v)=(u,J^*v). \end{equation}
Since $J$ is symmetric it already follows from this adjoint relation that  $(u,(J^*-J)v)=0$ for all $u$. Since there is no isotropic vector in $V$ we conclude that $J^*=J$ (since $v$ is arbitrary). Hence the fundamental symmetry operator $J$ is self-adjoint in $(V,(,))$.

The final claim is clear since $J^*=J=J^{-1}$, so that $JJ^*=I$.
\end{proof}

\begin{Lemma} (Schwarz inequality)\, Let $(H,(,))$ be a Hilbert space. Then for any vectors $u,v\in H$ there holds
$$ |(u,v)| \leq \|u\|\,\|v\|,$$
where $\|u\|=\sqrt{(u,u)}$. A similar result holds in a Krein space, $(H,[,])$; that is, for $u,v\in H$ there holds
$$ |[u,v]| \leq \|u\|\,\|v\|.$$
\label{lem0}
\end{Lemma}
For a proof see \cite{ikl}.

We consolidate some of these results in a theorem.
\begin{Theorem}
\label{thm4}
Let $(V,[,])$ be a Krein space and $(V,(,))$ its Hilbert majorant. Then there is a fundamental symmetry $J:V\rightarrow V$ such that for all $u,v \in V$,
\begin{equation}
\label{ind5}
 (u,v)=[Ju,v],
\end{equation}
and
\begin{equation}
 [u,v]=(Ju,v),
\end{equation}
where $J$ is an involution (i.e., $J^2=I$), $J$ is self-adjoint in $(V,())$, $J$ is unitary and $\|J\|=1$.
\end{Theorem}
So, in practice, given a Krein space $(V,[,])$ one can write down its canonical decomposition
\begin{equation}
 V=H^+[+]H^-
\end{equation}
from which one deduces the form of its orthoprojectors $P_\pm$. After this we can define the fundamental symmetry, $J=P_+-P_-$, and the corresponding positive definite inner product, $(,)$ via \eqref{ind5} in Theorem~\ref{thm4}.

\begin{Example}
The space $(\mathbb{R}^2,[,])$ with $[u,v]=(Bu,v)$ where $$B=\left(
\begin{array}{cc}
                                                              1 & 0 \\
                                                                0 & -1 \\
                                                              \end{array}
                                                            \right)$$
                                                            admits a decomposition in $[,]-$orthogonal subspaces $H^+$ and $H^-$ (spanned by the standard basis vectors $e_1$, $e_2$ where $[e_1,e_1]=1$ and $[e_2,e_2]=-1$).

                                                            By Theorem~\ref{thm4}, $(u,v)=[Ju,v]=(BJu,v)$ so that $((BJ-I)u,v)=0$ for all $u,v \in V=\mathbb{R}^2$.
                                                            Therefore, $(BJ-I)u=0$ (as $V$ has no isotropic vector).                                                      Since $u$ is arbitrary $BJ=I$ or $J=B^{-1}=B$ i.e., $J=B$.

%&&&&&&&&&&&&&&&&&&&&&&&& Feb 9, 4:17 p.m.

\end{Example}
\begin{Example}
Consider the Pontryagin space $(\mathbb{R}^4,[,])$ with the {\it Lorentz metric}, $[,]$, induced by the matrix $$B=\left(
                                                        \begin{array}{cccc}
                                                          1 & 0 & 0 & 0 \\
                                                          0 & -1 & 0 & 0 \\
                                                          0 & 0 & -1 & 0 \\
                                                          0 & 0 & 0 & -1\\
                                                        \end{array}
                                                      \right),$$

via the relation $[u,v]=(Bu,v)$. Its fundamental symmetry $J$ is given by setting
$[u,v]=(Ju,v)$ for all $u,v$. Hence $J=B$.
\end{Example}
One can also work ``backwards" to define a Krein space as the following theorem implies.
\begin{Theorem}
Let $(V,(,))$ be a Hilbert space and $J:V\rightarrow V$ a fundamental symmetry. Then the sesquilinear hermitian form $[,]$, defined by $[u,v]=(Ju,v)$ for all $u,v \in V$ defines a Krein space $(V,[,])$.
\end{Theorem}
\begin{Remark}
In the special case where $J$ is a positive definite fundamental symmetry the Krein space is reduced to a Hilbert space. Generally speaking, in a Krein space $J$ is generally indefinite (i.e., $(Ju,u)$ could be positive, negative, or zero without $u$ being necessarily zero).
 \end{Remark}
 \begin{Remark}
The inner product $(,)$ that is used to define the indefinite inner product in Examples~\ref{exam7} and ~\ref{exam8} is not necessarily the same as the one that is used to define the inner product of the Hilbert majorant space. In other words, if $B=B^*$ is given and $[u,v]=(Bu,v)$ for all $u,v \in V$ is some indefinite inner product,  then there may exist a possibly different inner product $<,>$ such that $[u,v]=<Ju,v>$ where $J$ is now a fundamental symmetry. One of the reasons for this is that $B^2\neq I$ necessarily, yet $B^2 = I$ would be required if the two inner products were identical.
\end{Remark}

In the same vein it is clear that since $(Bu,v) = <Ju,v>$ for all $u,v$ we can replace $u$ by $Ju$ so that the two inner products, $(,)$ and $<,>$, are actually related by the equality
\begin{equation*}
<u,v>=(BJu,v).
\end{equation*}
In addition, if $JB=BJ$ then $(JB)^*=(BJ)^*=J^*B^*=JB$ thus $JB$, and so $BJ$ is symmetric. We know from the elementary theory of Banach spaces that all norms on a finite dimensional space are (topologically) equivalent, thus even though the inner products $(,)$, $<,>$, are different, the norms defined by them are topologically equivalent. We give here a simple example illustrating this fact.

\begin{Example}
 Let $V=\mathbb{R}^2$ and define an indefinite inner product on $V$ by setting $[u,v]=(Bu,v)$ where $$B=\left(\begin{array}{cc}
 1 & 2 \\
 2 & 1 \\                                                                                      \end{array}
\right).$$
Since $B$ has real distinct eigenvalues of opposite sign it follows that the form
$[,]$ is indefinite. Furthermore, for $u={\rm col}\  (u_1,u_2)$ etc. we have
$$[u,v] = u_1v_1+2u_2v_1+2u_1v_2+u_2v_2.$$
Corresponding to the eigenvalues $\lambda_1=-1$, $\lambda_2=3$ of $B$ are the eigenvectors $f_1= {\rm col}\  (1,-1)$ and $f_2={\rm col}\  (1,1)$.

%%%%%%%%%% Feb 9, 2011, 9: 29 p.m. %%%%%%%%%%%%%%%%%%%%%%%%%%%%%%%%%%%%%

We now  define  $H^+:={\rm span}\ \{f_2\}$ and $H^-:={\rm span}\ \{f_1\}$. Since $[f_2,f_2]=6$ and $[f_1,f_1]=-2$, it follows that $H^+:={\rm span}\ \{f_2\}$ and $H^-:={\rm span}\ \{f_1\}$ are respectively positive and negative subspaces of $V$. Furthermore, since for positive and negative vectors $u^\pm \in H^\pm$ we have $[u^+,u^+]=(Bu^+,u^+)=6\alpha^2>0$ if $u^+=\alpha f_2$, and $-[u^-,u^-]=-(Bu^-,u^-)= 2\alpha^2>0,$ if $u^-=\alpha f_1$, it follows that $(H^+,[,])$ and $(H^-,-[,])$ are positive definite subspaces.

On the other hand, $[\alpha f_2,\beta f_1]=(\alpha B f_2,\beta f_1)=\alpha\overline{\beta}\,(Bf_2,f_1)=0.$    Hence the subspaces $H^+$ and  $H^-$ are orthogonal relative to $[\ ,\ ]$ i.e., $\mathbb{R}^2=H^+[+]H^-$, and so $(\mathbb{R}^2, [,])$ is a Pontryagin space.

Now that $(\mathbb{R}^2,[,])$ is an indefinite space, by Theorem~\ref{thm4}  there is a (positive definite) inner product $<,>$ such that $(\mathbb{R}^2, <,>)$ is a Hilbert (or Banach) space, with $<,>$ defined via a fundamental symmetry $J$ by
\begin{equation}\label{eq218} <u,v>=[Ju,v]=[u^+,v^+]-[u^-,v^-]\end{equation}
where, as usual, $u=u^++u^-.$

Observe that if $u={\rm col}\ (u_1,u_2)$ then
$u^+= {\rm col}\ (\frac{u_1+u_2}{2}, \frac{u_1+u_2}{2})$ and $u^-={\rm col}\ (\frac{u_1-u_2}{2}, \frac{u_2-u_1}{2})$, so that $u=u^++u^-$ with a similar calculation for $v$. Substituting the latter into \eqref{eq218} and simplifying yields
$$<u,v>=[u^+,v^+]-[u^-,v^-]=(Bu^+,v^+)-(Bu^-,v^-)=u_1v_2+2u_1v_1+2u_2v_2+u_2v_1,$$
that is,
  \begin{equation}\label{eq219}
  <u,v>=u_1v_2+2u_1v_1+2u_2v_2+u_2v_1.
  \end{equation}
  In order to find $J$ we note that
  $$J=J^*=\left(\begin{array}{cc}
               a & b \\
               b & c
             \end{array}\right)$$
  implies that
$$<u,v>=[Ju,v]=(BJu,v)=(a+2b)u_1v_1+(b+2c)u_2v_1+(2a+b)u_1v_2+(c+2b)u_2v_2,$$ for any choice of $u, v$, i.e.,
  \begin{equation}\label{eq220}
  <u,v>=(a+2b)u_1v_1+(b+2c)u_2v_1+(2a+b)u_1v_2+(c+2b)u_2v_2.
  \end{equation}
  Comparing \eqref{eq220} and \eqref{eq219} yields the unique solution $a=c=0$ and $b=1$ so that the required fundamental symmetry is
   $$J=\left(
   \begin{array}{cc}
    0 & 1 \\
    1 & 0 \\
    \end{array}
    \right).$$
Note that  $J=J^*$, $J^2=1$ (so that $JJ^*=I$), and that $<,>$ and $(,)$ are different inner products on the same space. In addition, $J\neq B$ and $BJ=JB$. The resulting norms defined by $\sqrt{<,>}$ and $\sqrt{(,)}$ are topologically equivalent.
\label{ex13}
\end{Example}

  In the next example we show that the canonical decomposition of a Krein space into an $[,]$-orthogonal direct sum of two (Hilbert) spaces is not unique.

 \begin{Example}
 Let $(\mathbb{R}^2,[,])$ be a Krein space with the indefinite inner product $[u,v]=(Bu,v)$, $u,v \in \mathbb{R}^2$
 where $B$ is the matrix of Example~\ref{ex13}. In that example we produced a canonical decomposition for the Krein space $\mathbb{R}^2$ in the form $\mathbb{R}^2=H^+[+]H^-$ where $H^+:={\rm span}\ \{f_2\}$ and $H^-:={\rm span}\ \{f_1\}$ where $f_1= {\rm col}\  (1,-1)$ and $f_2={\rm col}\  (1,1)$.
 Here we are looking for a possibly different canonical decomposition.

 It suffices to find a different basis of $\mathbb{R}^2$ consisting of ``positive" and ``negative" vectors relative to the inner product $[,]$. For example, if we choose $H^+:={\rm span}\ \{g_2\}$ and $H^-:={\rm span}\ \{g_1\}$ where $g_1= {\rm col}\  (1,-2)$ and $g_2={\rm col}\  (0,1)$, then $R^2=H^+[+]H^-$ is also a canonical decomposition.

 In this case the orthoprojectors $P_\pm$ are given by
 \begin{equation}
          P_+=\left(
                                                     \begin{array}{cc}
                                                       0 & 0 \\
                                                       2 & 1 \\
                                                     \end{array}
                                                   \right),\quad\quad P_-=\left(
                                                     \begin{array}{cc}
                                                       1 & 0 \\
                                                       -2 & 0 \\
                                                     \end{array}
                                                   \right).
\end{equation}
so that the corresponding fundamental symmetry operator is given by
          \begin{equation}
          J=P_+-P_-=\left(
          \begin{array}{cc}
            -1 & 0 \\
             4 & 1 \\
             \end{array}
             \right).
          \end{equation}

While $J^2=I$ is satisfied, note that $J$ is not symmetric because according to Lemma~\ref{lem2} (with $(,)$ replaced by $<,>$) the fundamental symmetry is to be symmetric relative to the inner product induced by the orthoprojectors $P_\pm$, i.e., $<u,v>:=[u^+,v^+]-[u^-,v^-]$ that is while $J$ is not our usual ``symmetric" matrix, $BJ$ (or $JB$) {\it is} symmetric as
         $$BJ=\left(
                \begin{array}{cc}
                  1 & 2 \\
                  2 & 1 \\
                \end{array}
              \right)\left(
                \begin{array}{cc}
                  -1 & 0 \\
                  4 & 1 \\
                \end{array}
              \right)=\left(
                \begin{array}{cc}
                  7 & 2 \\
                  2 & 1 \\
                \end{array}
              \right)=(BJ)^*.$$
                  \end{Example}
\begin{Remark}
Sometimes we say that $J$ is ``$B-$symmetric". The reason for this notation is since $u = u^+ + u^-$ we have $Ju = u^+-u^-$, by definition, so that
\begin{eqnarray}
[Ju,v] &=& [u^+-u^-,v^++v^-] \nonumber\\
&=&[u^+,v^+]-[u^-,v^-]\nonumber\\
 &=& <u,v>\quad ( {\rm by\ definition\ of\ <,>)}\label{eq221} \\
 &=&[u^++u^-,v^+-v^-]\nonumber\\
 &=& [u,Jv],\nonumber\\\nonumber
\end{eqnarray}
or $J$ is symmetric relative to $[,]$; but this latter inner-product is defined in terms so of $B$, so in this sense $J$ is $B$-symmetric.

Next, we observe that $J$ is also $<,>$-symmetric, that is $<u,Jv>=<Ju,v>$. This is because $<u,v>=[Ju,v]$ (see \eqref{eq221}) so that
\begin{eqnarray*}
<u,Jv> &=& [Ju,Jv]=[u^+-u^-,v^+-v^-] \\
&=&[u^+,v^+]+[u^-,v^-]=[u^++u^-,v^++v^-]\\
&=& [u,v] \\
&=& <Ju,v>.
\end{eqnarray*}
\end{Remark} % modified Feb. 11, 2011, 1:01 p.m.

%%%%%%%%%%%%%%%%%%%%%% Feb 10, 2011, 4:03 p.m. %%%%%%%%%%%%%%%%%%%%%%%%%%%%

\begin{Example}
Let $B$ be the matrix of Example~\ref{ex13}. Then $\mathbb{R}^2=H^+[+]H^-$, where $H^+={\rm span}\ \{{\rm col}\,(1,0)\}$ and $H^-={\rm span}\ \{{\rm col}\,(2,-1)\}$ and $[H^+,H^-]=0$, where $[u,v]=(Bu,v)$ as usual. Since the orthoprojectors are $$P_+=\left(
                                                                           \begin{array}{cc}
                                                                             1 & 2 \\
                                                                             0 & 0 \\
                                                                           \end{array}
                                                                         \right),\quad\quad\quad P_-=\left(
                                                                           \begin{array}{cc}
                                                                             0 & -2 \\
                                                                             0 & 1 \\
                                                                           \end{array}
                                                                         \right).$$
the fundamental symmetry becomes
                                                                         $$J=P_+-P_-=\left(
                                                                           \begin{array}{cc}
                                                                             1 & 4 \\
                                                                             0 & -1 \\
                                                                           \end{array}
                                                                         \right).$$

We see that, generally speaking, $J$, $P_+$ and $P_-$ are not ``symmetric" in the ordinary sense, i.e., they are not $(,)$-symmetric, but $BJ$, $BP_+$ and $BP_-$ are symmetric from the theory. Recall that when the theory is applied to this example, we get
\begin{eqnarray}
\label{eq222}
[u,v]&=& (Bu,v)\\
\label{eq223}
<u,v>&=& [Ju,v].
\end{eqnarray}
For example, that $BJ$ is symmetric follows from \eqref{eq222}-\eqref{eq223}, the symmetry of $B$ and the argument
\begin{equation*} (BJu,v)=[Ju,v]=[u,Jv]=<u,v>=[u,Jv]=(Bu,Jv)=(u,BJv).
\end{equation*}

This can also be verified directly since
                                                                         $$BJ=\left(
                                                                           \begin{array}{cc}
                                                                             1 & 2 \\
                                                                             2 & 1 \\
                                                                           \end{array}
                                                                         \right)\left(
                                                                           \begin{array}{cc}
                                                                             1 & 4 \\
                                                                             0 & -1 \\
                                                                           \end{array}
                                                                         \right)=\left(
                                                                           \begin{array}{cc}
                                                                             1 & 2\\
                                                                             2 & 7 \\
                                                                           \end{array}
                                                                         \right)=(BJ)^*$$
                                                                         $$BP_+=\left(
                                                                           \begin{array}{cc}
                                                                             1 & 2 \\
                                                                             2 & 1 \\
                                                                           \end{array}
                                                                         \right)\left(
                                                                           \begin{array}{cc}
                                                                             1 & 2 \\
                                                                             0 & 0 \\
                                                                           \end{array}
                                                                         \right)=\left(
                                                                           \begin{array}{cc}
                                                                             1 & 2\\
                                                                             2 & 4 \\
                                                                           \end{array}
                                                                         \right)=(BP_+)^*,$$ and $$BP_-=\left(
                                                                           \begin{array}{cc}
                                                                             1 & 2 \\
                                                                             2 & 1 \\
                                                                           \end{array}
                                                                         \right)\left(
                                                                           \begin{array}{cc}
                                                                             0 & -2 \\
                                                                             0 & 1 \\
                                                                           \end{array}
                                                                         \right)=\left(
                                                                           \begin{array}{cc}
                                                                             0 & 0\\
                                                                             0 & -3 \\
                                                                           \end{array}
                                                                         \right)=(BP_-)^*.$$
So, even though $P_\pm$ and $J$ are not symmetric they can be shown to be symmetric relative to $<,>$.
\end{Example}

%%%%%%%%%%%%%%%%%%%%%%% Re-insert Section 3, acc. to SH revisions dated Feb 14, 2011 %%%%%%%%%%%%%%%%%%%%%%%%%%%

\section{Heisenberg's Uncertainty Principle}
%\textbf{Heisenberg's Uncertainty Principle:} Rewritten by SH Feb 14, 2011
In this section we give a brief history of the uncertainty principle. The uncertainty principle, one of the most characteristic consequences of quantum mechanics, was first formulated by Heisenberg (see \cite{Heisenberg44}) for two {\it conjugate quantum variables} (i.e., they can be defined so that they are Fourier transform duals of one another.)

According to Heisenberg's uncertainty principle, the product of the uncertainties in the measurement of two conjugate quantum variables is at least of the order of Planck's constant $h$. To derive the uncertainty relation, Heisenberg considered the notion of a wave packet. A ``wave packet" is a generally moving disturbance whose amplitude is noticeable only in a bounded region which changes its size and shape; in other words, a spreading disturbance. For example, consider the function
\begin{equation}
f(x)=\int_{-\infty}^\infty g(k)\,e^{ikx}\,dk\nonumber
\end{equation}
whose real part,
\begin{equation}
\int_{-\infty}^\infty g(k)\cos(kx)\,dk,\nonumber
\end{equation}
is a linear superposition of waves with wave number $k$, or wave length $\lambda ={2\pi}/{k}$. If $g(k)$ is given by the Gaussian
\begin{equation}
g(k)=e^{-\alpha(k-k_0)^2}\nonumber
\end{equation}
where $\alpha > 0$, then
\begin{eqnarray*}
f(x)&=&\int_{-\infty}^\infty  g(k)e^{i(k-k_0)x}e^{ik_0x}\,dk\\
&=&e^{ik_0x}\int_{-\infty}^{\infty} e^{ik^{\prime}x}e^{-\alpha k^{\prime2}}\,dk^{\prime}\quad (k^{\prime}=k-k_0)\\ &=&e^{ik_0x}\int_{-\infty}^{\infty} e^{-\alpha(k^{\prime}-(ix/2\alpha))^2}e^{-x^2/4\alpha}\,dk^{\prime}\\
&=& \sqrt{\frac{\pi}{\alpha}}e^{ik_0x}e^{-x^2/4\alpha}.
\end{eqnarray*}
Thus,
\begin{equation}
|f(x)|^2=\frac{\pi}{\alpha}e^{-x^2/2\alpha}.\nonumber
\end{equation}
This is a function with a maximum at $x=0$ and width, [see e.g., \cite{gasiorowicz44}, p. 28], of the order $2\sqrt{2\alpha}$. If we compare this width with the width of $|g(k)|^2=e^{-2(k-k_0)^2\alpha}$ i.e., with ${2}/{\sqrt{2\alpha}}$, we see that
 \begin{equation}
\Delta k\Delta x\sim \frac{2}{\sqrt{2\alpha}}.2\sqrt{2\alpha}=4\nonumber
\end{equation}
i.e., the product of the widths of the two functions is independent of the parameter $\alpha$. It is a general property of any two functions that are Fourier transforms of each other that the wider the function the sharper is its Fourier transform and vice versa. Therefore for any two conjugate operators we have, as a result,
\begin{equation}
\Delta k\Delta x\geq O(1) > 0.\nonumber
\end{equation}
In the case of the conjugate variables $p, x$ of quantum mechanics, this becomes
\begin{equation}
\Delta p\Delta x\geq \frac{\hbar}{2}.\nonumber
\end{equation}
(because $p=\hbar k$ for a plane wave.) i.e., there is a limitation on the accuracy of simultaneous measurements of two quantum variables that are conjugate of each other. This is a general feature of wave packets and is a consequence of the assumption that if the position of a particle is known with a certain accuracy $\triangle x$, then it can be visualized as a wave packet in that position with a width equal to $\triangle x$, i.e., the idea of looking at notions of position and momentum of particles at atomic scales as superpositions of simple waves $e^{\pm ikx}$.

Later on in 1929, E. U. Condon pointed out that in the general case, when two quantum variables are not conjugate of each other, the uncertainty relation may not hold in its usual form and it is therefore important to have a more general formulation, \cite{condon44}.

Inspired by this comment of Condon, Robertson \cite{Robertson44} presented such a generalization in 1929 for any two non-commuting self-adjoint operators $A$ and $B$ on (the Hilbert space) $L^2(\mathbb{R})$ whose commutator $[A,B]=iC$, where $C$ is now a general (albeit at least symmetric or self-adjoint) operator.

In \cite{Robertson44}, for a quantum system with a normalized wave function $\psi$, Robertson defined  a mean value $A_0$ (the expectation value of $A$ in the state $\psi$), associated to such a self-adjoint operator $A$ by
\begin{equation}
A_0=\int\overline{\psi}A\psi\, d\tau \nonumber
\end{equation}
where $d\tau$ is the element of configuration space and the integral is extended over all of euclidean space. It is easily seen that $A_0$ is real since $A$ is self-adjoint, i.e., since
\begin{equation}
\int\overline{\phi}A\psi d\tau= \int\overline{\psi}A\phi\, d\tau \nonumber
\end{equation}
for all $\phi$ and $\psi$ (indeed, this is simply the relation $(A\psi,\varphi) = (\psi,A\varphi)$ using the standard inner product on $L^2(\mathbb{R})$.)

In statistics, the uncertainty in the value of $A$ is defined as
 \begin{equation}
(\triangle A)^2=\int\overline{\psi}(A-A_0)^2\psi\, d\tau, \nonumber
\end{equation}
i.e., the square root of the mean of the deviation of $A$ from its mean value $A_0$.

To prove the uncertainty relation in a specific case Robertson considers the quantum mechanical variables $A(q,p)$ and $B(q,p)$ that are linear in the momenta $(p_x,p_y,p_z)$, i.e., $A=a+a_xp_x+a_yp_y+a_zp_z$ where $p_x=({\hbar/i})\cdot{\partial}/{\partial x}$, etc. where $a, a_x,a_y,a_z$ are functions of position. Then, using the self-adjointness of $A$, Robertson calculated the uncertainty in $A$ in this case to be
 \begin{equation}
(\triangle A)^2=\int|(A-A_0)\psi|^2 d\tau. \nonumber
\end{equation}
The Schwarz inequality
\begin{equation}
\label{schi}
|\int(f_1g_1+f_2g_2)d\tau|^2\leq [\int(f_1\bar{f}_1+f_2\bar{f}_2)d\tau][\int(g_1\bar{g}_1+g_2\bar{g}_2)d\tau]
\end{equation}
is applied taking
\begin{equation}
\bar{f}_1=(A-A_0)\psi=f_2, \quad g_1=(B-B_0)\psi=-\bar{g}_2. \nonumber
\end{equation}
Integrating the left hand side of \eqref{schi} by parts then yields
 \begin{equation}
\Delta A\Delta B\geq\frac{1}{2}\left |\int\bar{\psi}(AB-BA)\psi d\tau \right|. \nonumber
\end{equation}
Since $iC=AB-BA = [A,B]$ by hypothesis, we get
\begin{equation}
\Delta A\Delta B\geq\frac{1}{2}\left |\int\bar{\psi}C\psi d\tau\right |, \nonumber
\end{equation}
or, using the inner product notation,
\begin{equation}
\label{eq300}
\sigma(A)(\psi)\sigma(B)(\psi)\geq\frac{1}{2}|(C\psi,\psi)|
\end{equation}
where $|(C\psi,\psi)|$ is the modulus of the expectation value of $C$, \cite{Robertson44}.

If we take $A=\delta$ and $B=P=-i\hbar\frac{d}{dx}$, then $[\delta,P]=i\hbar$, i.e., $C=\hbar I$, then the general formula \eqref{eq300} reduces to
\begin{equation}
(\triangle \delta)(\triangle P)\geq\frac{1}{2}\hbar,
\end{equation}
which is the original form of Heisenberg's uncertainty relation between the position and momentum variables (here $\triangle$ stands for $\sigma$, the standard deviation).

He concluded that the Heisenberg uncertainty principle is just a specific case of a more general relation between the standard deviations of two non-commuting operators in $L^2(\mathbb{R})$.

In 1930, Von Neumann (cf., \cite{Neumann44}) proved an abstract Hilbert space version of Heisenberg's uncertainty relation for any two non-commuting operators in a space of square integrable functions by relying on the proof by Robertson \cite{Robertson44}. He started with $[A,B]=aI$ and the fact that $a$ must be pure imaginary, because
 \begin{equation}
[A,B]^*=-[A,B], \quad [A,B]^*=(aI)^*=\bar{a}I \Rightarrow \bar{a}I=-aI \Rightarrow \bar{a}=-a.\nonumber
\end{equation}
Then, for a vector $\varphi$ in the Hilbert space $(H,(,))$,
\begin{equation}
ia\|\varphi\|^2=(i[A,B]\varphi,\varphi)=2\,{\rm Im}\,(A\varphi,B\varphi),
\end{equation}
from which
$$
\|\varphi\|^2\leq\frac{2}{|a|}|{\rm Im}\,(A\varphi,B\varphi)|r
\leq\frac{2}{|a|}|(A\varphi,B\varphi)|\leq\frac{2}{|a|}\|A\varphi\|\|B\varphi\|,$$
where, once again, the Schwarz inequality is used in the previous equation. Therefore, assuming that  $\|\varphi\|=1$, this yields
\begin{equation}
\label{eq301}
\|A\varphi\|\|B\varphi\|\geq\frac{|a|}{2}.
\end{equation}
From this and the fact that for $\sigma$ and $\rho$ real, the quantities $A-\sigma\,I$ and $B-\rho\, I$ are self-adjoint and obey the same commutation relation as $A$ and $B$, \eqref{eq301} can be applied to $A-\sigma\,I$ and $B-\rho\,I$, with the choice $\sigma=(A\varphi,\varphi)$ and $\rho=(B\varphi,\varphi)$. This then gives
\begin{equation}\label{eq302}
\|(A-\sigma I)\varphi\|\|(B-\rho I)\varphi\|\geq\frac{|a|}{2}.
\end{equation}
On the other hand,
\begin{equation}
\|(A-\sigma\,I)\varphi\|=\|(A-(A\varphi,\varphi)\,I)\varphi\|
=(A\varphi-(A\varphi,\varphi)\varphi,A\varphi-(A\varphi,\varphi)\varphi)^{1/2}\nonumber
\end{equation}
\begin{equation}
=\{(A\varphi,A\varphi)-2(A\varphi,\varphi)^2+(A\varphi,\varphi)^2\}^{1/2}
=\{(A\varphi,A\varphi)-(A\varphi,\varphi)^2\}^{1/2}\nonumber
\end{equation}
\begin{equation}
=\{(A^2\varphi,\varphi)-(A\varphi,\varphi)^2\}^{1/2}=\{E(A^2)-E(A)^2\}^{1/2}:=\sigma(A)\nonumber
\end{equation}
where $E(A)$ is expectation value of $A$ in the state $\varphi$ and $\sigma(A)$ is the uncertainty in the measurement of $A$. Therefore
\begin{equation}
\|(A-\sigma\,I)\varphi\|:=\sigma(A), \quad {\rm similarly,} \quad \|(B-\rho\,I)\varphi\|:=\sigma(B)\nonumber
\end{equation}
are standard deviations. Using this notation the inequality \eqref{eq302} gives the uncertainty relation
\begin{equation}
\sigma(A)\sigma(B)\geq\frac{|a|}{2}.\nonumber
\end{equation}
This is essentially Von Neumann's argument. A further state-dependent generalization of the uncertainty relation was obtained by Schr\"{o}dinger in \cite{Schrodinger44} for any two self-adjoint operators $A$ and $B$ when the system is in a state $\varphi$, where $\varphi$ is in a Hilbert space. He considered the fact that the product of two self-adjoint operators is, in general, not self-adjoint but yet this  product can be decomposed into a sum of a self-adjoint and a ``skew-self-adjoint" ($i$ $\times$ self-adjoint) operator, i.e.,
\begin{equation}
AB=\frac{AB+BA}{2}+\frac{AB-BA}{2}
\end{equation}
where the first term is self-adjoint and the second term is skew self-adjoint. Next,
\begin{equation}
\label{eq303}
|(B\varphi,A\varphi)|^2
=(AB\varphi,\varphi)=(\frac{AB+BA}{2}\varphi,\varphi)^2+|(\frac{AB-BA}{2}\varphi)|^2.
\end{equation}
On the other hand, the Schwarz inequality gives
\begin{equation}
\label{eq304}
|(B\varphi,A\varphi)|^2\leq\|B\varphi\|^2\|A\varphi\|^2.
\end{equation}
Now \eqref{eq303} and \eqref{eq304} together give
\begin{equation}
\|B\varphi\|^2\|A\varphi\|^2\geq(\frac{AB+BA}{2}\varphi,\varphi)^2+|(\frac{AB-BA}{2}\varphi)|^2.
\end{equation}
Reasoning as after \eqref{eq301} he proceeds to derive the following inequality for the standard deviations of the self-adjoint operators $A$ and $B$ when in the state $\varphi\in H$:
\begin{equation}
\label{eq224}
\sigma_{A}\sigma_{B}\geq
\sqrt{\frac{1}{4}|([A,B]\varphi,\varphi)|^{2}+\frac{1}{4}|(\{A-(A\varphi,\varphi)I,B-(B\varphi,\varphi)I\}\varphi,\varphi)|^{2}},
\end{equation}
where $[A,B]=AB-BA$ and $\{A,B\}=AB+BA$ are the commutator and anti-commutator of $A$ and $B$ respectively. This relation includes the commutation part of the formula that was derived by Robertson \cite{Robertson44} as well as a new anti-commutator part. It is easily seen that omission of the anti-commutator part in \eqref{eq224} yields the weaker inequality of Robertson (and so Heisenberg.)

Landau and Peierls \cite{landau44}, \cite{de broglie44} pointed out a form of uncertainty for the cases where speeds of particles approach the speed of light and so relativistic corrections need to be taken into account. In the following sections we derive general uncertainty relations for non-commuting self-adjoint operators in Krein space. We notice that in the particular case where the fundamental symmetry is equal to the identity operator, $J=I$, our general uncertainty relation reduces to the one by Schr\"{o}dinger \cite{Schrodinger44} for Hilbert space. It is also shown that in the special case where the operators are the linear momentum and position of a particle, it reduces to a stronger inequality than the usual Heisenberg uncertainty relation. We also present a generalization of the uncertainty principle where the commutator of the self-adjoint operators is not a multiple of the identity but a combination of the identity and one of the two operators.

%%%%%%%%%%%%%%%%%%%%%%% Feb.11, 2011, 1:24 p.m.

%*******************************************************************************************************************************************
% Deleted Section 3 due to its length and assimilated it into Section 4.
%******************************************************************************************************************************************

\section{Generalizations of the Uncertainty Principle}
\label{sec31}
In this section we proceed to extend the uncertainty principle due to Heisenberg-Robertson-Von-Neumann-Schr\"{o}dinger to operators that are not necessarily self-adjoint in a Hilbert space but are self-adjoint in a space with a possibly indefinite inner-product. As a result of this endeavor, all concepts must now be redefined including the notion of a standard deviation. This generalization shows that the base operators, $A, B$ of the preceding section need {\it not} be self-adjoint in a Hilbert space in order for an uncertainty relation to hold.

In what follows, operators $A, B, \ldots$ on a given complex Hilbert space, $(H,(,))$ may be bounded or unbounded. If they are bounded, then there is no need for domain considerations; the domain of the commutator/anti-commutator may be taken to be the whole Hilbert space, $H$. Thus, in the general case where at least one of $A, B$ is an unbounded operator, we need to put some restrictions on the domains of the commutator $[A, B]$. For our purposes we do not wish $A$ nor $B$ to have a common element $x$ in their null space (as then $[A,B]x=0$ and some proofs degenerate). Generally speaking the two operators need to satisfy one or more of the following conditions as a bare minimum.

\begin{equation*}
\left\{
\begin{array}{ll}
1, & \hbox{Dom ([A,B])=Dom(A)$ \cap $ Dom(B) is dense in H}, \\
2, & \hbox{Ran(B) $\subseteq$ Dom(A), Ran(B) is dense in $H$},\\
3, & \hbox{Ran(A) $\subseteq$ Dom(B), Ran(A) is dense in $H$,}\nonumber
\end{array}
\right.
\end{equation*}\\
where Dom$(A)$ and Ran$(A)$ stand for domain and range of the operator $A$ respectively. In the main application, $H$ is a weighted $L^2$-space with an indefinite inner product. In this case, we will assume that the (definite or indefinite) inner product produces a topology in which the space $C^{\infty}_0$ defined as the space of infinitely differentiable functions having compact support in $\mathbb{R}$, is dense in $H$. This is so if, for example, the Lebesgue measure of the set of points at which the weight vanishes is zero. As a result we can assume that the domains of $A$ and $B$ contain $C^{\infty}_0$ at least.

%%%%%%%%%%%%%%%% Feb. 11, 2011, 8:35 p.m.

\section{Uncertainty in Krein space}
In this section we introduce some of the theory of linear operators in Krein space that will be necessary for us in the formulation of our main results.
\begin{Lemma}\label{lemm2}
Let $(H,[,])$ be a Krein space and $A$ a self-adjoint operator on $(H,[,])$. Then $JA$ is self-adjoint in the Hilbert majorant space, $(H,(,))$.
\end{Lemma}
\begin{proof}
See \cite{ikl}.
\end{proof}
\begin{Lemma}
If $A$ is self-adjoint in the Hilbert space $(H,(,))$, then $JA$ is self-adjoint in the Krein space, $(H,[,])$ where, as usual, $[f,g]=(Jf,g)$.
\end{Lemma}
\begin{proof}
See \cite{ikl}.
\end{proof}
Let $A$ be a linear operator on a Hilbert space $(H,(,))$. We say that $A$ is {\it skew-hermitian} if $A$ and $A^*$ have the same domain and $A^*=-A$. An operator $A$ is said to be {\it $J$-skew-hermitian} or {\it $[,]$-skew-hermitian} on the Krein space $(H,[,])$, if $A^+ = -A$.

\begin{Lemma}\label{lemm3}
Let $A$, $B$ be operators in a Krein space $(H,[,])$ that satisfy Lemma~\ref{lem1}, parts (1)-(2)-(3), with equality.
\begin{enumerate}
\item If $A$, $B$ are self-adjoint in $(H,[,])$ then $[A,B]$ is $J$-skew-hermitian and $\{A,B\}$ is self-adjoint in the Krein space,
\item If $J$ is a fundamental symmetry and if $A$, $B$ are self-adjoint operators on $(H,[,])$,  then ${\rm Re}\, [J[A,B]\varphi,\varphi] =0,$
 that is, $[J[A,B]\varphi,\varphi]$ is purely imaginary,
\item Finally, if $A$ is $J$-skew-hermitian then $(JA\varphi,\varphi) = [A\varphi,\varphi]$ is purely imaginary.
\end{enumerate}
\end{Lemma}
\begin{proof} This is a simple calculation. Note that
$[A,B]^+=(AB-BA)^+=(AB)^+-(BA)^+=B^+A^+-A^+B^+=BA-AB=-[A,B]$. The part about the anti-commutator is similar.

Next, since $[J[A,B]\varphi,\varphi] = ([A,B]\varphi,\varphi)$ (by Theorem~\ref{thm4}) and $[A,B]$ is skew-hermitian,
\begin{eqnarray*}
([A,B]\varphi,\varphi) &=& (\varphi,[A,B]^*\varphi) \\
&=& - (\varphi,[A,B]\varphi)\\
&=& - \overline{([A,B]\varphi,\varphi)},
\end{eqnarray*}
so that ${\rm Re}\, [J[A,B]\varphi,\varphi] =0,$ as required.

Now, let $A$ be $J$-skew-hermitian. Then $$[A\varphi,\varphi] = [\varphi,A^+\varphi] = -\overline{[A\varphi,\varphi]},$$
and so $[A\varphi,\varphi]$ must be purely imaginary.
\end{proof}

The next lemma is important in what follows and deals with a useful property of sesquilinear hermitian forms.
\begin{Lemma}
\label{lemm4}
Let $A$ be any linear operator on a vector space $V$, let $[,]$ be a sesquilinear hermitian form on $V\times V$, and $\sigma \in \mathbb{C}$. Then, for $\varphi \in V$ we have
$$[(A-\sigma I)\varphi,(A-\sigma I)\varphi] = [A\varphi, A\varphi] - 2\,{\rm Re}\, \{ \sigma [\varphi,A\varphi]\}+|\sigma|^2 [\varphi, \varphi].$$
\end{Lemma}

\begin{proof} We expand the left side using the definition of the form (cf., Definition~\ref{def1}). Thus,
\begin{eqnarray*}
[(A-\sigma I)\varphi, (A-\sigma I)\varphi]&=& [A\varphi, (A-\sigma I)\varphi]+[-\sigma \varphi,(A-\sigma I)\varphi]\\
&=&[A\varphi,A\varphi] +[A\varphi,-\sigma \varphi] + \overline{[A\varphi-\sigma\varphi,-\sigma\varphi]}\\
&=& [A\varphi,A\varphi] - \overline{\sigma}\overline{[\varphi, A\varphi]} + \overline{[A\varphi,-\sigma\varphi]} + \overline{[-\sigma\varphi,-\sigma\varphi]}\\
&=&[A\varphi,A\varphi] - \overline{\sigma}\overline{[\varphi, A\varphi]} -\sigma\, {[\varphi,A\varphi]} - \overline{\sigma}\,\overline{[\varphi,-\sigma\varphi]}\\
&=& [A\varphi,A\varphi] - \overline{\sigma}\overline{[\varphi, A\varphi]} -\sigma\,{[\varphi,A\varphi]} - \overline{\sigma}\,{[-\sigma\varphi,\varphi]}\\
&=&[A\varphi,A\varphi] - \overline{\sigma}\overline{[\varphi, A\varphi]} -\sigma\, {[\varphi,A\varphi]} - \overline{\sigma}(-\sigma)\,{[\varphi,\varphi]}\\
&=&[A\varphi,A\varphi] - \overline{\sigma}\overline{[\varphi, A\varphi]} -\sigma\, {[\varphi,A\varphi]} + |\sigma|^2 [\varphi,\varphi]\\
&=& [A\varphi, A\varphi] - 2\,{\rm Re}\, \{ \sigma [\varphi,A\varphi]\}+|\sigma|^2 [\varphi, \varphi],
\end{eqnarray*}
as required.
\end{proof}
\begin{Corollary}
 \label{cor2}
 Let $A$ be any linear operator on a Hilbert space $(H,(,))$, $\varphi \in {\rm Dom}\,(A).$ Then, for any $\sigma \in \mathbb{C}$ we have
  $$\|A\varphi\|^2 - 2\,{\rm Re}\, \{ \sigma (\varphi,A\varphi)\}+|\sigma|^2 \|\varphi\|^2 \geq 0,$$
  where, by definition, $\|u\|^2 = (u,u)$, for any $u\in H$.
\end{Corollary}
\begin{proof} The proof is clear since, by Lemma~\ref{lemm4}, the left side is equal to $\|(A-\sigma I)\varphi\|^2 \in \mathbb{R}$.
\end{proof}
\begin{Corollary}
\label{cor1}
Let $A$ be any self-adjoint operator on a Hilbert space $(H,(,))$, $\varphi \in H$, $\sigma \in \mathbb{C}.$
Then $$\|(A-\sigma\,I)\varphi \|^2 = \|A\varphi\|^2 -2 ({\rm Re}\, \sigma)\,(\varphi, A\varphi)+ |\sigma|^2 \|\varphi\|^2.$$
\end{Corollary}
\begin{proof} Since $A$ is self-adjoint, $(\varphi, A\varphi)$ is always real. The relation is now clear by Corollary~\ref{cor2}.
\end{proof}

Our next result is our first main result. It establishes an uncertainty principle for self-adjoint operators in a Krein space that includes the Heisenberg-Robertson-Schr\"{o}dinger-Von Neumann relations as a special case.
\begin{Theorem}
\label{thm1}
Let $A$, $B$ be self-adjoint operators in a Krein space $(H,[,])$ that satisfy Lemma~\ref{lem1}, parts (1)-(2)-(3) with equality, and let $J$ be a fundamental symmetry. As usual let $(H,(,))$ be the Hilbert majorant space. Then for a state $\varphi \in H$, normalized in $(H,(,))$ by setting $(\varphi,\varphi)=1$, we have

$\displaystyle \sigma_J^{2}(A)(\varphi)\sigma_J^{2}(B)(\varphi)\geq$
\begin{equation}
\label{eq238}
\left \{(J\frac{\{A,B\}}{2}\varphi,\varphi)
-(JA\varphi,\varphi)(JB\varphi,\varphi)(2-(J\varphi,\varphi))\right \}^{2}+ \left
|(J\frac{[A,B]}{2}\varphi,\varphi)\right |^{2},
\end{equation}
where
\begin{eqnarray}
\label{jsdd}
 \sigma_J^{2}(A)(\varphi)&=&(A\varphi,A\varphi)-2(JA\varphi,\varphi)\,{\rm Re}\,(A\varphi,\varphi)+(JA\varphi,\varphi)^{2}
 \end{eqnarray}
 is the squared {\bf $J$-standard deviation} of the operator $A$ in the state $\varphi$ in the Krein space $(H,[,])$.
 \end{Theorem}

{\bf Note:} One can easily see that for $J=I$ we get the real number $\sigma_J=\sigma$, where $\sigma^{2}=(A^{2}\varphi,\varphi)-(A\varphi,\varphi)^{2}$ is the usual \textit{standard deviation} in the measurement of the observable $A$ in the state $\varphi \in H$. Thus, \eqref{jsdd} may be thought of as the generalized standard deviation of a self-adjoint operator in a Krein space.

\begin{proof} (Theorem~\ref{thm1})
Write the product of the operators $A$ and $B$ in the form
 \begin{equation*}
 AB=\frac{AB+BA}{2}+\frac{AB-BA}{2}=\frac{\{A,B\}}{2}+\frac{[A,B]}{2}.
 \end{equation*}
 The hypotheses imply that $\{A,B\}$ is $[,]$-self-adjoint and $[A,B]$ is $[,]$-skew-hermitian. Since $A$ is $[,]$-self-adjoint we get that $JA$ is self-adjoint so that
  \begin{eqnarray*}
 (B\varphi,JA\varphi)&=&(JAB\varphi,\varphi)\\
 &=&(J\frac{\{A,B\}}{2}\varphi,\varphi)+(J\frac{[A,B]}{2}\varphi,\varphi).\nonumber
 \end{eqnarray*}
Thus,
\begin{eqnarray}
\left |(B\varphi,JA\varphi)\right |^{2}
&=&\left |\frac{1}{2}\,(J\{A,B\}\varphi,\varphi)+\frac{1}{2}\,(J[A,B]\varphi,\varphi)\right |^{2}\label{eq225}
\end{eqnarray}
However, since $\{A,B\}$ is $[,]$-self-adjoint it is the case that $J\{A,B\}$ is self-adjoint and so the first term in \eqref{eq225} is necessarily real. On the other hand, the second term is purely imaginary (i.e., it has its real part equal to zero, by Lemma~\ref{lemm3}\,(3) with $A$ there replaced by the commutator). Hence,
\begin{eqnarray}\label{eq226}
|(B\varphi,JA\varphi)|^{2} &=& \frac{1}{4}\,(J\{A,B\}\varphi,\varphi)^2+\frac{1}{4}\,|(J[A,B]\varphi,\varphi) |^{2}.
\end{eqnarray}
Use of the Schwarz inequality (Lemma~\ref{lem0}) we find
\begin{equation}\label{eq227}
|(B\varphi,JA\varphi)|^{2} \leq \|B\varphi\|^{2}\|JA\varphi\|^{2}\leq\|B\varphi\|^{2}\|A\varphi\|^{2}.
\end{equation}
Combining \eqref{eq226}-\eqref{eq227} gives \begin{equation}\label{eq228}\frac{1}{4}\,(J\{A,B\}\varphi,\varphi)^2+\frac{1}{4}\,|(J[A,B]\varphi,\varphi) |^{2} \leq \|A\varphi\|^{2}\|B\varphi\|^{2}.
\end{equation}

%%%%%%%%%%%%%%%%%% Feb. 12, 2011, 3:55 p.m....top of p.25 in draft

Observe that \eqref{eq228} necessarily holds if we replace $A$ by $A-\sigma I$ and $B$ by $B-\rho I$ with $\sigma$ and $\rho$ being arbitrary real numbers.

So, we choose $\sigma=[A\varphi,\varphi]=(JA\varphi,\varphi)$ and $\rho=[B\varphi,\varphi]=(JB\varphi,\varphi)$ i.e., the expectation values of the self-adjoint operators $A$ and $B$ in the Krein space (so that $\sigma$ and $\rho$ are necessarily real numbers).

Referring to \eqref{eq228} we see that we need to compute the two commutators of the translated operators. We first consider the anti-commutator:

$\displaystyle \{A-(JA\varphi,\varphi)I,B-(JB\varphi,\varphi)I\}$
\begin{eqnarray}
&=& (A-(JA\varphi,\varphi)I)(B-(JB\varphi,\varphi)I)\nonumber\\
&& +(B-(JB\varphi,\varphi)I)(A-(JA\varphi,\varphi)I)\nonumber\\
&=& AB- A(JB\varphi,\varphi)-(JA\varphi,\varphi)B+(JA\varphi,\varphi)(JB\varphi,\varphi)\nonumber\\
&& +BA-(JA\varphi,\varphi)B-(JB\varphi,\varphi)A+(JB\varphi,\varphi)(JA\varphi,\varphi) I\nonumber\\
&=&\{A,B\}-2(JB\varphi,\varphi)A-2(JA\varphi,\varphi)B+2(JA\varphi,\varphi)(JB\varphi,\varphi) I.\label{eq229}
\end{eqnarray}
Next, we consider the commutator:

$\displaystyle [A-(JA\varphi,\varphi)I,B-(JB\varphi,\varphi)I]$
\begin{eqnarray}
&=& (A-(JA\varphi,\varphi)I)(B-(JB\varphi,\varphi)I)\nonumber\\
&& -(B-(JB\varphi,\varphi)I)(A-(JA\varphi,\varphi)I)\nonumber\\
&=& AB -A(JB\varphi,\varphi)-(JA\varphi,\varphi)B+(JA\varphi,\varphi)(JB\varphi,\varphi)\nonumber\\
&& -[BA-(JA\varphi,\varphi)B-(JB\varphi,\varphi)A+(JB\varphi,\varphi)(JA\varphi,\varphi))]\nonumber\\
&=& AB-BA\nonumber\\
&=&[A,B]. \label{eq230}
\end{eqnarray}
Using \eqref{eq229} we get

$\displaystyle \frac{1}{2}\,(J\{A-(JA\varphi,\varphi)I,B-(JB\varphi,\varphi)I\}\varphi,\varphi)$
\begin{eqnarray}
&=& \frac{1}{2}\,(J\{A,B\}\varphi,\varphi)-(JB\varphi,\varphi)(JA\varphi,\varphi)\nonumber\\
&& -(JA\varphi,\varphi)(JB\varphi,\varphi)+ (JA\varphi,\varphi)(JB\varphi,\varphi)(J\varphi,\varphi)\nonumber\\
&=&  \frac{1}{2}\,(J\{A,B\}\varphi,\varphi)
-(JA\varphi,\varphi)(JB\varphi,\varphi)\{2-(J\varphi,\varphi)\}.\label{eq231}
\end{eqnarray}

Similarly, we find
\begin{equation}
\frac{1}{2}\,(J[A-(JA\varphi,\varphi)I,B-(JB\varphi,\varphi)I]\varphi,\varphi)
= \frac{1}{2}\,(J[A,B]\varphi,\varphi).\label{eq232}
\end{equation}

Therefore, as a result of the shift we get \eqref{eq231}-\eqref{eq232} for the two commutators of the translated operators.

On the other hand, since $\sigma = (JA\varphi,\varphi) \in \mathbb{R}$, we can use Corollary~\ref{cor2} to find that
\begin{eqnarray}
\|(A-(JA\varphi,\varphi)I)\varphi\|^{2} &=& (A\varphi,A\varphi)-2(JA\varphi,\varphi)\,{\rm Re}\,(A\varphi,\varphi)+(JA\varphi,\varphi)^{2} \label{eq235} \\
&\geq & 0,\label{eq236}
\end{eqnarray}
where we normalize $\varphi$ so that $(\varphi,\varphi)=1$ in Corollary~\ref{cor2}.

We conclude from Corollary~\ref{cor2} that $$\sigma_J^{2}(A)(\varphi) = (A\varphi,A\varphi)-2(JA\varphi,\varphi)\,{\rm Re}\,(A\varphi,\varphi)+(JA\varphi,\varphi)^{2}$$
is well-defined as a real number and this quantity defines the square of the $J$-standard deviations referred to in \eqref{jsdd}.

Similarly, using the same notation,
\begin{equation}\label{eq234} \sigma_J^{2}(B)(\varphi) = \|(B-(JB\varphi,\varphi)I)\varphi\|^{2}.
\end{equation}
Combining the results \eqref{eq231}, \eqref{eq232}, \eqref{eq235} and \eqref{eq234} with \eqref{eq228} gives
us the desired lower bound,

$\displaystyle \sigma_J^{2}(A)(\varphi)\sigma_J^{2}(B)(\varphi)\geq$
\begin{equation}
\{(J\frac{\{A,B\}}{2}\varphi,\varphi)-(JA\varphi,\varphi)(JB\varphi,\varphi)[2-(J\varphi,\varphi)]\}^{2}+
|(J\frac{[A,B]}{2}\varphi,\varphi)|^{2}.
\end{equation}
\end{proof}

\begin{Remark}
  The only restriction on the operators $A$ and $B$ in order for the inequality \eqref{eq228} to hold is that $A$ and $B$ are self-adjoint in the Krein space subject to some basic hypotheses on their adjoints (cf., Lemma~\ref{lem1}).
\end{Remark}

%%%%%%%%%%%%%%%%%%%% Feb 13, 2011, 6:50 p.m.

\begin{Remark}
 Note that for $J=I$ (i.e., when the Krein space is a Hilbert space) and the operator $A$ is $(,)$-self-adjoint, the $J$-standard deviation defined in \eqref{jsdd} becomes
\begin{eqnarray*}
\sigma_J^{2}(A)(\varphi)&=&(A\varphi,A\varphi)-2(JA\varphi,\varphi)\,{\rm Re}\,(A\varphi,\varphi)+(JA\varphi,\varphi)^{2}\\
&=&(A\varphi,A\varphi)-2(A\varphi,\varphi)(A\varphi,\varphi)+(A\varphi,\varphi)^{2}\\
&=&(A\varphi,A\varphi)-(A\varphi,\varphi)^{2}\\
&=&(A^{2}\varphi,\varphi)-(A\varphi,\varphi)^{2},
\end{eqnarray*}
which is the usual ``standard deviation" described in terms of operators. Assuming that $\varphi$ is normalized to unity in $H$ then \eqref{eq238} reduces to
\begin{equation}
\sigma^{2}(A)(\varphi)\sigma^{2}(B)(\varphi)
\geq\left[(\frac{\{A,B\}}{2}\varphi,\varphi)-(A\varphi,\varphi)(B\varphi,\varphi)\right]^{2}+
\left|(\frac{[A,B]}{2}\varphi,\varphi)\right|^{2}\end{equation} which is the inequality obtained by Schr\"{o}dinger, \cite{Schrodinger44}. So, we see that our general (Krein space) inequality, \eqref{eq238}, includes the Schr\"{o}dinger inequality in the specific case where the Krein space is a Hilbert space.
\end{Remark}

In the following theorems the effect of a generalized commutation rule in Krein space is considered in order to prove different versions of the uncertainty principle.
\begin{Theorem}\label{thm7}
Let $[A,B]=aJ$ where $a$ is purely imaginary and $A$, $B$ are self-adjoint operators in a Krein space $(H,[,])$ having $J$ as a fundamental symmetry. Then for $\varphi\in H$,
\begin{equation}
\label{eq237}
\sigma_J(A)(\varphi)\sigma_J(B)(\varphi)\geq \frac{|a|}{2}|(J\varphi,\varphi)|,\nonumber
\end{equation}
where
\begin{equation}
\sigma_J(A)(\varphi)=\{(A\varphi,A\varphi)-2(JA\varphi,\varphi)\,{\rm Re}\,(A\varphi,\varphi)+(JA\varphi,\varphi)^{2}\}^{1/2} \nonumber
\end{equation}
etc., is the $J$-standard-deviation in the measurement of the observable $A$.
\end{Theorem}

\begin{proof} First we note that $a$ must be purely imaginary in order for $aJ$ in the commutation hypothesis to be skew-hermitian in $H$ (since $J$ is self-adjoint). Next, for suitable $\varphi$ in the domain of $A$ and $B$ we have,
\begin{eqnarray*}
a(J\varphi,\varphi) &=& (aJ\varphi,\varphi)\\
&=&([A,B]\varphi,\varphi) \\
&=&((AB-BA)\varphi,\varphi)\\
&=&(AB\varphi,\varphi)-(BA\varphi,\varphi)
\end{eqnarray*}
\begin{eqnarray*}
\therefore |a||(J\varphi,\varphi)|\leq|(AB\varphi,\varphi)|+|(BA\varphi,\varphi)|.
\end{eqnarray*}
Next
\begin{eqnarray*}
(AB\varphi,\varphi)=[B\varphi,AJ\varphi].\\
\therefore |(AB\varphi,\varphi)|\leq \|B\varphi\|\|A\varphi\|
\end{eqnarray*}by Schwarz inequality and since $\|J\|=1$.\\
Similarly
\begin{eqnarray*}
|(BA\varphi,\varphi)|\leq\|A\varphi\|\|B\varphi\|.
\end{eqnarray*}
Since
\begin{eqnarray*}
|a||(J\varphi,\varphi)|&\leq&|(AB\varphi,\varphi)|+|(BA\varphi,\varphi)|\\
&\leq&2\|A\varphi\|\|B\varphi\|\\
\therefore |(J\varphi,\varphi)|\leq \frac{2}{|a|}\|A\varphi\|\|B\varphi\|
\end{eqnarray*}
Hence,
\begin{equation}
\label{eq239}
\|A\varphi\|\|B\varphi\|\geq\frac{|a|}{2}|(J\varphi,\varphi)|.
\end{equation}

Now consider the commutator, $[A-\sigma I,B-\rho I]$, where $\sigma, \rho$ are real or complex quantities. A straightforward calculation now shows that the commutator is translation invariant, i.e.,
 $$[A-\sigma I,B-\rho I]=[A,B]=aJ,$$
i.e., $A-\sigma I$ and $B-\rho I$ obey the same commutation rule as does $A,B$. Also $A-\sigma I$ and $B-\rho I$ are $J$-self-adjoint (i.e., self-adjoint with respect to $[,]$) only if $\sigma, \rho$ are real (see Lemma~\ref{lem1}\, (6)).

 Now, since $A$, $B$ are $J$-self-adjoint the quantities $(B\varphi,\varphi), (A\varphi,\varphi)$ are not necessarily real, so we cannot choose $\rho, \sigma$ as we would like. Instead, we choose $\rho=(JB\varphi,\varphi)$ and $\sigma=(JA\varphi,\varphi)$. The latter quantities being real, the operators $A-\sigma I$ and $B-\rho I$ are now $J$-self-adjoint and these enjoy the same commutation rule as does $A,B$. Hence we can derive the relation,
\begin{equation}
\label{eq240}
\|(A-\sigma I)\varphi\|\|(B-\rho I)\varphi\|\geq\frac{|a|}{2}|(J\varphi,\varphi)|.
\end{equation}
Since $\sigma, \rho \in \mathbb{R}$ and $(\varphi,\varphi)=1$, we use  Lemma~\ref{lemm4} with the form $[,]$ being replaced by the inner product, $(,)$, to get,
\begin{eqnarray*}
\frac{|a|}{2}|(J\varphi,\varphi)|&\leq& \|(A-\sigma I)\varphi\|\|(B-\rho I)\varphi\|\\
 &=&\|(A-(JA\varphi,\varphi) I)\varphi\|\|(B-(JB\varphi,\varphi) I)\varphi\|\\
 &=& \{(A\varphi,A\varphi)-2(JA\varphi,\varphi)\,{\rm Re}\,(A\varphi,\varphi)+(JA\varphi,\varphi)^{2}\}^{1/2}\\
 && \times\{(B\varphi,B\varphi)-2(JB\varphi,\varphi)\,{\rm Re}\, (B\varphi,\varphi)+(JB\varphi,\varphi)^{2}\}^{1/2}\\
 &=&\sigma_J(A)(\varphi)\sigma_J(B)(\varphi),
  \end{eqnarray*}
 since ${\rm Re}\,(\varphi,A\varphi) = {\rm Re}\,(A\varphi,\varphi)$ with a similar relation for $B$. The result follows.
\end{proof}

The proof of the following result is actually a special case of the techniques used in proving Theorem~\ref{thm7}.

\begin{Theorem}\label{thm8}
Let $[A,B]=aJ$ where $a$ is purely imaginary and $A$, $B$ are self-adjoint operators in a Krein space $(H,[,])$ having $J$ as a fundamental symmetry. Then for $\varphi\in H$,
\begin{equation}
\label{eq237}
\sigma_J(A)(\varphi)\sigma_J(B)(\varphi)\geq \frac{|a|}{2}|(J\varphi,\varphi)|,\nonumber
\end{equation}
where
\begin{equation}
\sigma_J(A)(\varphi)=\{(A\varphi,A\varphi)-({\rm Re}\,(A\varphi,\varphi))^2\}^{1/2} \nonumber
\end{equation}
etc., is another type of $J$-standard-deviation in the measurement of the observable $A$.
\end{Theorem}

\begin{proof} The proof proceeds exactly as in the previous theorem except for the choice of $\sigma, \rho$. In this case, we choose $\sigma = {\rm Re}\, (A\varphi,\varphi)$ and $\rho = {\rm Re}\, (B\varphi,\varphi)$ in \eqref{eq240}.The remaining argument is similar.
\end{proof}

\begin{Remark} Given a state $\varphi$, the number of choices of $\sigma, \rho$ in \eqref{eq240} is basically infinite (see also Lemma~\ref{lemm4}). Hence, the techniques used in proving Theorem~\ref{thm7} and Theorem~\ref{thm8} can be used to provide general forms of uncertainty principles using different ``standard deviations" on the left.
\end{Remark}

A result in the same spirit as the one obtained in Theorem~\ref{thm8} can be obtained in Hilbert space, i.e., when the operators $A, B$ are self-adjoint in a Hilbert space, but where their  commutator is dependent upon an unspecified fundamental symmetry as in said theorem.

\begin{Theorem}\label{thm9}
Let $[A,B]=aJ$ where $a$ is purely imaginary, $J$ is a fundamental symmetry, and $A$, $B$ are self-adjoint operators in a Hilbert space $(H,(,))$. Then for $\varphi\in H$,
\begin{equation}
\label{eq237}
\sigma_J(A)(\varphi)\sigma_J(B)(\varphi)\geq \frac{|a|}{2}|(J\varphi,\varphi)|,\nonumber
\end{equation}
where
\begin{equation}
\sigma_J(A)(\varphi)=\{(A\varphi,A\varphi)-(A\varphi,\varphi)^{2}\}^{1/2} \nonumber
\end{equation}
etc., is the usual standard-deviation in the measurement of the observable $A$ in the state $\varphi$.
\end{Theorem}

\begin{proof} As before, $a$ must be purely imaginary in order for $aJ$ in the commutation hypothesis to be skew-hermitian in $H$ (since $J$ is self-adjoint). Similarly we can show that (as in Theorem ~\ref{thm7})
\begin{eqnarray*}
a(J\varphi,\varphi) &=& (aJ\varphi,\varphi)\\
&=&([A,B]\varphi,\varphi) \\
&=&2i\,{\rm Im}\,(B\varphi,A\varphi),
\end{eqnarray*}
and
\begin{eqnarray*}
|(J\varphi,\varphi)|&=&|\frac{2i}{a}{\rm Im}\,(B\varphi,A\varphi)|\\
&\leq& \frac{2}{|a|}\|B\varphi\|\|A\varphi\|,
\end{eqnarray*}
so that,
\begin{equation}
\label{eq239}
\|A\varphi\|\|B\varphi\|\geq\frac{|a|}{2}|(J\varphi,\varphi)|.
\end{equation}

The translation invariance of the commutator, $[A-\sigma I,B-\rho I]$, where $\sigma, \rho$ are real or complex quantities is used with the choices $\rho=(B\varphi,\varphi)$ and $\sigma=(A\varphi,\varphi)$, both now being real in \eqref{eq240}. The remainder of the proof proceeds with minor changes and so is omitted.
\end{proof}

%%%%%%%%%%%%%%%%%% Feb 14, 2011, 12:01 p.m.

\section{The case of an operator dependent commutator}
In this section we consider the problem of an operator dependent commutator. First we derive some fundamental properties of commutators of self-adjoint operators in a Krein space.
\begin{Lemma}
 \label{lemm5}
Let $A$ and $B$ be self-adjoint operators in a Krein space, $(H,[,])$ with Hilbert majorant space $(H,(,))$. Let $J$ be a fundamental symmetry, and assume that $[J,A]=[J,B]=0$.
Then
$$[JA, JB] = [A,B] = - [A,B]^+,$$
and $$[A,B]^* =  -[JA^*, JB^*].$$
\end{Lemma}
\begin{proof} Since $A$ and $B$ are $J$-self-adjoint operators, it follows by Lemma~\ref{lemm2} that $JA$, $JB$ are self-adjoint in the Hilbert majorant space, $(H,(,))$. So, by the usual property of commutators, $[JA,JB]^*= - [JA,JB]$. But $$[JA,JB]= JAJB-JBJA = J^2AB-J^2BA = AB-BA=[A,B].$$ On the other hand, $[A,B] = - [A,B]^+$ since $A,B$ are self-adjoint in the Krein space.

Next, the commutativity assumptions between $J, A, B$ imply that $[J,A^*]=0$ and $[J,B^*]=0$, since $J=J^*$. In addition,
\begin{eqnarray*}
[A,B]^* &=& (AB-BA)^*\\
&=& -(A^*B^*-B^*A^*)\\
&=& -(AJ)^*(BJ)^*+(BJ)^*(AJ)^*\\
&=& -[(AJ)^*, (BJ)^*]\\
&=& -[JA^*, JB^*].
\end{eqnarray*}
\end{proof}
\begin{Theorem}
\label{thm10}
Let $[A,B]=a\,(I+\beta B^{2})$ where $A$ and $B$ are self-adjoint operators in a Krein space, $(H,[,])$, $J$ is a fundamental symmetry, $[J,A]=[J,B]=0$, $a$ is pure imaginary and $\beta \in \mathbb{R}$. Then we have
\begin{eqnarray}
 \sigma_J(A)(\varphi)\sigma_J(B)(\varphi)&\geq& \frac{|a|}{2}\,[(1+ \beta\,\{{\rm Re}\,(B^2\varphi,\varphi) - 2(JB\varphi,\varphi)\,{\rm Re}\,(B\varphi,\varphi) + (JB\varphi,\varphi)^2\})^2  \nonumber \\
 && + (\beta\,\{{\rm Im}\,(B^2\varphi,\varphi) -2(JB\varphi,\varphi)\,{\rm Im}\,(B\varphi,\varphi)\})^2 ]^{1/2}.\label{eq247}
 \end{eqnarray}
 where
\begin{equation*}
 \sigma_J^{2}(A)(\varphi)=(A\varphi,A\varphi)-2(JA\varphi,\varphi)\,{\rm Re}\,(A\varphi,\varphi)+(JA\varphi,\varphi)^{2}.
 \end{equation*}
\end{Theorem}
\begin{proof} First, we note that $a$ must be pure imaginary since, by Lemma~\ref{lemm5}, $[A,B]$ is $J$-skew-hermitian (i.e., if $a$ has a non-zero real part, then the assumption on the commutators is impossible).

We use Theorem~\ref{thm1}. Substituting $[A,B]=a(I+\beta B^{2})$ into \eqref{eq238} yields,

$\displaystyle \sigma_J^{2}(A)(\varphi)\sigma_J^{2}(B)(\varphi)\geq$
\begin{equation}\geq\{(J\frac{\{A,B\}}{2}\varphi,\varphi)
-(JA\varphi,\varphi)(JB\varphi,\varphi)[2-(J\varphi,\varphi)]\}^{2}+
|(J\frac{a(I+\beta B^{2})}{2}\varphi,\varphi)|^{2}
\end{equation}

From the commutation rule and the commutativity of $J$ with $A$ and $B$ we get,
\begin{eqnarray*}
(a(1+\beta B^{2})\varphi,\varphi)&=&((AB-BA)\varphi,\varphi)\\
&=&(J^2AB\varphi,\varphi)-(J^2BA\varphi,\varphi)\\
&=&(JAJB\varphi,\varphi)-(JBJA\varphi,\varphi)\\
&=&(JB\varphi,JA\varphi)-(JA\varphi,JB\varphi)\\
&=&(JB\varphi,JA\varphi)-\overline{(JB\varphi,JA\varphi)}\\
&=& 2i\,{\rm Im}\,(JB\varphi,JA\varphi)\\
&=&2i\,{\rm Im}\,(B\varphi,A\varphi).
\end{eqnarray*}
Therefore,
\begin{equation}
\label{eq241}
(a\varphi,\varphi)+(a\beta B^{2}\varphi,\varphi)=2i\,{\rm Im}\,(B\varphi,A\varphi).
\end{equation}

%%%%%%%%%%%%%%%%%%%%%%%% Feb. 15, 2011, 11:17 p.m.%%%%%%%%%%%%%%%%%%%%%%%%%%%%%%%%%%%%%%%%%%%%%%%

Now
\begin{eqnarray*}
2|{\rm Im}\,(B\varphi,A\varphi)| &=& |(a\varphi,\varphi)+(a\beta B^{2}\varphi,\varphi)|\\
&=& |a||(\varphi,\varphi)+\beta\, (B^{2}\varphi,\varphi)|,
\end{eqnarray*}
so that choosing $\varphi$ with $\|\varphi\|=1$ we get, by Schwarz's inequality,
\begin{eqnarray*}
2\|B\varphi\|\|A\varphi\|&\geq& 2|(B\varphi,A\varphi)|\\
& \geq& 2|{\rm Im}\,(B\varphi,A\varphi)|\\
&=& |a||1+\beta\, (B^{2}\varphi,\varphi)|.
\end{eqnarray*}
Therefore,
\begin{equation}
\label{eq242}
\|B\varphi\|\|A\varphi\|\geq\frac{|a|}{2}|1+\beta\,(B^{2}\varphi,\varphi)|.
\end{equation}
Note that $B^2$ is self-adjoint in the Krein space and so $(B^{2}\varphi,\varphi)$ is not necessarily real. As before we can show that, for any $\sigma, \rho \in \mathbb{C}$,
$$[A-\sigma\,I,B-\rho\,I]=[A,B].$$
Next, for $\sigma, \rho \in \mathbb{R}$, the operators $A-\sigma\,I,B-\rho\,I$ are self-adjoint in the Krein space as well, and $J$ commutes with both of these for any constants $\sigma, \rho$. Consequently, applying the above argument to these translated operators we find that \eqref{eq242} becomes, for any real $\sigma$ and real $\rho$,
 \begin{equation}
 \label{eq243}
 \|(B-\rho\,I)\varphi\|\|(A-\sigma\,I)\varphi\|
 \geq\frac{|a|}{2}|1+\beta\,((B-\rho\,I)^2\varphi,\varphi)|.
 \end{equation}
Now, choosing the real quantities $\sigma, \rho$ as before, that is, $\sigma=(JA\varphi,\varphi), \rho=(JB\varphi,\varphi)$ and writing $$((B-\rho\,I)^2\varphi,\varphi)={\rm Re}\,((B-\rho\,I)^2\varphi,\varphi) + i\,{\rm Im}\,((B-\rho\,I)^2\varphi,\varphi),$$
\eqref{eq243} becomes,
\begin{eqnarray}
\label{eq245}
\|(B-(JB\varphi,\varphi)\,I)\varphi\|\,\|(A-(JA\varphi,\varphi)\,I)\varphi\|\nonumber\\
\geq\frac{|a|}{2}|1+\beta\,((B-(JB\varphi,\varphi)\,I)^2\varphi,\varphi)|.
\end{eqnarray}
Now, a straightforward calculation gives
$$((B-(JB\varphi,\varphi)\,I)^2\varphi,\varphi) = (B^2\varphi,\varphi) -2(JB\varphi,\varphi)(B\varphi,\varphi) + (JB\varphi,\varphi)^2,$$
and since $(JB\varphi,\varphi) \in \mathbb{R}$ this becomes
\begin{eqnarray*}
%\label{eq244}
((B-(JB\varphi,\varphi)\,I)^2\varphi,\varphi) &=& \{{\rm Re}\,(B^2\varphi,\varphi) -2(JB\varphi,\varphi)\,{\rm Re}\,(B\varphi,\varphi) \\
&& +(JB\varphi,\varphi)^2\}  \\
&& +i \{{\rm Im}\,(B^2\varphi,\varphi) -2(JB\varphi,\varphi)\,{\rm Im}\,(B\varphi,\varphi)\}.
\end{eqnarray*}
Thus, from \eqref{eq245},

$\displaystyle \frac{|a|}{2}\,|1+\beta\,((B-(JB\varphi,\varphi)\,I)^2\varphi,\varphi)|$
\begin{eqnarray*}
&=& \frac{|a|}{2}\, \bigg |1+ \beta\,\{{\rm Re}\,(B^2\varphi,\varphi) -
 2(JB\varphi,\varphi)\,{\rm Re}\,(B\varphi,\varphi)+  (JB\varphi,\varphi)^2\} \\
&& + i \beta\,\left \{{\rm Im}\,(B^2\varphi,\varphi) - 2(JB\varphi,\varphi)\,{\rm Im}\,(B\varphi,\varphi)\right \}\bigg |,\\
&=& \frac{|a|}{2}\,[(1+ \beta\,\{{\rm Re}\,(B^2\varphi,\varphi) - 2(JB\varphi,\varphi)\,{\rm Re}\,(B\varphi,\varphi)\\
&& + (JB\varphi,\varphi)^2\})^2 + (\beta\,\{{\rm Im}\,(B^2\varphi,\varphi) -2(JB\varphi,\varphi)\,{\rm Im}\,(B\varphi,\varphi)\})^2 ]^{1/2}
\end{eqnarray*}
so that \eqref{eq245} now becomes

$\displaystyle \|(B-(JB\varphi,\varphi)\,I)\varphi\|\,\|(A-(JA\varphi,\varphi)\,I)\varphi\|$
\begin{eqnarray}
&\geq& \frac{|a|}{2}\,[(1+ \beta\,\{{\rm Re}\,(B^2\varphi,\varphi) - 2(JB\varphi,\varphi)\,{\rm Re}\,(B\varphi,\varphi) + (JB\varphi,\varphi)^2\})^2  \nonumber \\
 && +(\beta\,\{{\rm Im}\,(B^2\varphi,\varphi) -2(JB\varphi,\varphi)\,{\rm Im}\,(B\varphi,\varphi)\})^2 ]^{1/2}. \label{eq246}
\end{eqnarray}
As usual,
\begin{equation}
 \|(B-(JB\varphi,\varphi)\,I)\varphi\|=\sigma_J(B)(\varphi),
 \end{equation}
and
\begin{equation}
 \|(A-(JA\varphi,\varphi)\,I)\varphi\|=\sigma_J(A)(\varphi).
 \end{equation}
Combining these with \eqref{eq246} we get the conclusion.
\end{proof}

\begin{Corollary}
\label{cor12}
When the Krein space is a Hilbert space (i.e., $J=I$) the result \eqref{eq247} becomes
\begin{equation}
\label{eq248}
 \sigma(A)(\varphi)\sigma(B)(\varphi)\geq\frac{|a|}{2}(1+\beta\sigma^2(B)(\varphi))
 \end{equation}
 where
 \begin{equation}
\sigma^2(B)=(B^2\varphi,\varphi)-(B\varphi,\varphi)^2
\end{equation}
is the square of the usual standard deviation of the self-adjoint operator, $B$.
\end{Corollary}

Assuming $\beta \geq 0$, we get from \eqref{eq248} that
$$\sigma(B)(\varphi)\geq\frac{|a|}{2\sigma(A)(\varphi)}(1+\beta\sigma^2(B)(\varphi)).$$
We iterate the preceding expression by replacing $\sigma (B)$ on the right by the whole expression on the right, e.g., define $f:\sigma \to \mathbb{R}$ by
$$f(\sigma) = \frac{|a|}{2\sigma(A)(\varphi)}(1+\beta\sigma^2),$$
and compute its orbit at $\sigma=0$. The sequence thus generated is defined by setting $f^n(0)=f(f^{n-1}(0))$, where $n=2,3,4,\ldots$, and $f^1(0)=f(0)$, whose $n$-th term gives the first $n$ terms in the expansion on the right.
In other words,
\begin{equation*}
 \sigma(B)(\varphi)
 \geq\frac{|a|}{2\sigma(A)(\varphi)}\left(1+\frac{\beta|a|^2}{4\sigma^{2}(A)(\varphi)}
 +\frac{\beta^2|a|^4}{8\sigma^4(A)(\varphi)}
+\frac{5\beta^3|a|^6}{64\sigma^{6}(A)(\varphi)} + \cdots\right),
 \end{equation*}
 or
 \begin{equation*}
 \sigma(A)(\varphi)\,\sigma(B)(\varphi)\geq\frac{|a|}{2}+\frac{\beta|a|^3}{8\sigma^{2}(A)(\varphi)}+\frac{\beta^2|a|^5}{16\sigma^4(A)(\varphi)}
+\frac{5\beta^3|a|^7}{128\sigma^{6}(A)(\varphi)} + \cdots,
 \end{equation*}
which is consistent with the form of a generalized uncertainty principle see, e.g., \cite{manic44}, \cite{x.Li44}, \cite{zhao44}, \cite{Ren44}.

We end our investigations with the following result that uses the self-adjointness of $B^*B$ in lieu of the non self-adjointness of $B^2$ as in Theorem~\ref{thm10}.

\begin{Theorem}
Let $[A,B]=a\,(I+\beta B^*B)$ where $A$ is a self-adjoint operator in a Krein space, $(H,[,])$, $B$ is closed,  $J$ is a fundamental symmetry, $[J,A]=[J,B]=0$, $a$ is pure imaginary and $\beta \in \mathbb{R}$. Then we have
\begin{equation}
 \sigma_J(A)(\varphi)\sigma_J(B)(\varphi)\geq\frac{|a|}{2}(1+\beta\sigma_J^2(B)(\varphi))
 \end{equation}
where
\begin{equation}
\sigma_J^{2}(B)(\varphi)=(B\varphi,B\varphi)-2(JB\varphi,\varphi){\rm Re}\,(B\varphi,\varphi)+(JB\varphi,\varphi)^{2}(\varphi,\varphi).
\end{equation}
\end{Theorem}
\begin{proof} The proof follows the same argument as that presented in  Theorem~\ref{thm10} with minor modifications, and so is omitted.
\end{proof}

%%%%%%%%%%%%%%%%%%%%%%%%%%%%%%%%%%%%%%% Feb. 17, 12:24 a.m.

%************************************************************************************************************************************************

\section{Conclusion}
\label{ch6}
We found various generalized versions of the Heisenberg uncertainty principle in indefinite inner-product spaces, specifically, Krein spaces. These new abstract forms of the uncertainty relations all reduce to the usual uncertainty principle in specific instances of the inner-product or the space, or the operators. While the notions of {\it uncertainty} through {\it standard deviations} in a Krein space setting are different than the corresponding notions in Hilbert space, our definitions coincide with the usual ones in the specific case of a Hilbert space. The importance of these generalizations lies more in the fact that the base operators are not conjugate and not even self-adjoint in the usual sense.

These results are in sharp contrast with early speculations by Condon (1927) on commutators and their relationship to uncertainty principles. Inspired by the latter paper this work is more concerned with the question:`` Can we extend the class of operators, $A$, $B$ and the spaces they are defined upon so as to {\it guarantee} an uncertainty relation?" We show that, indeed, there exist classes of non-self-adjoint operators on Hilbert spaces such that the non-vanishing of their commutator implies an uncertainty relation.

Another generalization of the uncertainty principle was given involving a general form of the commutators, that is, when the commutator of the two operators is not a multiple of the identity operator but a function of one of the two operators. This generalization was derived in Krein space for the operators commuting with the fundamental symmetry of the Krein space.

\bibliographystyle{plain}
\end{document}